\documentclass[12pt]{article}
\usepackage{apacite,latexsym,graphics,epsfig,amsmath,amssymb,tabularx,booktabs, color, mathtools,tikz}
\usepackage{natbib}
\usepackage{enumerate}
\usepackage{amsthm}
\usepackage{url}
\usepackage{fancyhdr,titling}
\usepackage[mathscr]{eucal}
\usepackage[utf8]{inputenc}
\usetikzlibrary{trees}

\tikzstyle{bag} = [align=center]

\usepackage[T1]{fontenc}
\theoremstyle{definition}
\newtheorem{definition}{Definition}[section]
\newtheorem{example}{Example}[section]
\theoremstyle{plain}
\newtheorem{theorem}{Theorem}[section]

\newtheorem{proposition}[theorem]{Proposition}

\newcommand{\tr}{^{\prime}}

\makeatletter
\newcommand{\oset}[2]{%
  {\mathop{#2}\limits^{\vbox to -.1\ex@{\kern-\tw@\ex@
   \hbox{\scriptsize #1}\vss}}}}
   
\newcommand{\keywords}[1]{%
  \par\vspace{0.5em}%
  \noindent\textbf{Keywords:} #1%
  \par\vspace{0.8em}%
}

\newcommand*\xbar[1]{%
  \hbox{%
    \vbox{%
      \hrule height 0.5pt 
      \kern0.5ex
      \hbox{%
        \kern-0.25em
        \ensuremath{#1}%
        \kern-0.1em
      }%
    }%
  }%
}

\makeatletter
\let\runauthor\@author
\let\runtitle 

\makeatletter
\newtheorem*{rep@theorem}{\rep@title}
\newcommand{\newreptheorem}[2]{%
\newenvironment{rep#1}[1]{%
 \def\rep@title{#2 \ref{##1}}%
 \begin{rep@theorem}}%
 {\end{rep@theorem}}}
\makeatother

\makeatletter
\newtheorem*{rep@example}{\rep@title}
\newcommand{\newrepexample}[2]{%
\newenvironment{rep#1}[1]{%
\def\rep@title{#2 \ref{##1}}%
 \begin{rep@example}}%
 {\end{rep@example}}}
\makeatother

\newreptheorem{theorem}{Theorem}
\newreptheorem{lemma}{Lemma}
\newrepexample{example}{Example}

\date{}

\begin{document}

\title{On the analysis of sequential designs without a specified number of observations}

\author{Anna Klimova \\
{\small{National Center for Tumor Diseases (NCT), Partner Site Dresden, and}}\\
{\small{Institute for  Medical Informatics and Biometry,}}\\ 
{\small{Technical University, Dresden, Germany} }\\
{\small \texttt{anna.klimova@nct-dresden.de}}\\
{}\\
\and 
Tam\'{a}s Rudas \\
{\small{Department of Statistics, E\"{o}tv\"{o}s Lor\'{a}nd University, Budapest, Hungary}}\\
{\small \texttt{trudas@elte.hu}}\\
}

\maketitle

\begin{abstract}
The paper focuses on sequential experiments whose outcomes, success or failure, can be expressed using a binomial distribution, and in which whether or not a further experiment is made depends on the prior outcomes. Examples include subsequent medical interventions until a target response is achieved.  
The maximum likelihood estimate of the treatment success probability is the ratio of the number of observed successes and the total number of times the treatment was administered. Despite a prespecified initial sample size, the total number of conducted interventions cannot be determined in advance, and the estimated total number of successes may not be equal to the one observed, so the estimated variance of success probability is not that of a binomial distribution. 
These results generalize to designs with a tree structure consisting of groups of the experiments with group-specific outcome distributions.  The multiplicative models entailed are studied using the tools from algebraic statistics and the theory of relational models. The estimates of success probabilities and their covariance are derived and shown to depend on the structural properties of the relevant exponential family. 
\end{abstract}


\section{Introduction}\label{sectionIntro}

Experimental designs aiming at investigating the distribution of an outcome of interest over subsequent interventions and across stratification factors occur  in many practical settings.
This paper works with the sequential designs that can be expressed using a tree, with a primary focus on  path dependent designs when sequential experiments are initiated conditionally on a previous response.  Such designs may occur in panel studies, life-event analyses,  educational training, and medical studies where the subsequent treatment may depend on the response to previously administered treatments.  For example,   in educational programs, one can use a certain training regimen until a trainee reaches proficiency (test-retest problems).   In oncology, there is often a first-line treatment is administered to newly diagnosed patients. A second-line treatment would be administered to all patients for whom the first treatment did not bring the desired effect. A third-line treatment would be applied if neither of the first two worked, and so on.  Different examples and relevant statistical hypotheses can be found in  \cite{JohnsonMay1995},  \cite{Agresti2002}, \cite{EdwardsAPFA}, among others.   This paper considers a class of multiplicative models whose parametric structure is implied by the data-generating process and can be  applied to the data collected under such path-dependent design. The data cross classifies participants according to their responses to the treatment received. But as later treatments are only administered to participants with particular responses, the cross classification is not complete, i.e., the sample space is not a full Cartesian product. These models, on one hand, generalize log-linear models \citep[cf.][]{Agresti2002, TRbook2018}, and relational models \citep*{KRD11}, and on the other hand,  are special cases of staged trees studied in  algebraic statistics \citep*[cf.][]{SmithAndersonTrees2008, GeorgenSmithTrees2018, SturmfelsDuarteMLE2021}.   The paper discusses the correct distributional assumptions,  obtains the maximum likelihood estimators of the model parameters, and discusses their properties not considered so far.  Because in path-dependent data, the total number of observations,  referred to as the \textit{exposure size} is not known in advance,   the sufficient statistics of the  maximum likelihood estimates and in their covariances are  shown to include an exposure-related factor that adjust the observed exposure size to the expected one. This adjustment does not exist for fixed i.e, non-path-dependent designs.  Further, the MLEs preserve the observed success rates but the expected number of successes is not equal to the observed value. Although the expressions for the maximum-likelihood estimators themselves are already known from \cite{SturmfelsDuarteMLE2021}, the properties of the sufficient statistics have not been previously investigated and, to our knowledge, no expression for the asymptotic covariance matrix of the MLE under a staged-tree model has been obtained.
These, together with the characteristics of the observed data preserved by the MLE, are, beyond mere model fit, key factors in deciding whether or not a particular statistical model is appropriate to describe the data generating mechanism.

The remainder of the paper is organized as follows.  Section \ref{sectionOneTrt} introduces the tree-based framework for sequential  experiment designs and proposes a multiplicative parameterization for testing the hypothesis of complete homogeneity of response distributions.  Such parameterizations entail a model class that can be represented using a staged tree with a single floret.  The maximum likelihood estimation in this model class,  the properties of the maximum likelihood estimators, including their asymptotic distribution,  are addressed in detail. Section \ref{SectionMultiple}  extends the discussion to the hypothesis of multi-group homogeneity,  when the treatments involved in the design are partitioned in several groups,  where homogeneity assumptions hold within each group individually.  The corresponding parameterizations and the models entailed by them can be expressed using a multi-floret staged tree. The properties of the maximum likelihood estimators in the multi-floret case are derived and shown to generalize the one-floret case.    Section \ref{SectionExample2} illustrates the results using real data, and Section \ref{sectionDiscussion} summarized the contribution of the paper.  Detailed proofs are presented in the Appendix.

\section{Multiplicative models for one-treatment sequential designs}\label{sectionOneTrt}

The focus of this paper is on investigating data that arise from a series of experiments.  Assume that each experiment is characterized by a random variable with a finite set of outcomes and their respective probabilities, and the experimental design has a tree structure, that is, starting from the first experiment (called root),  each subsequent experiment is performed conditionally on the outcome of the previous experiment, including the possibility of no further experiment being performed.  The tree structure implies that every node can be uniquely associated with the path leading to it.  A tree  in which all root-to-leaf paths are of the same length will be called complete, and a tree that  allows for the paths emanating from the same node and leading to a leaf be of different length will be called as path-dependent.  An outcome indicating that no further experiments follow is referred to as terminal, and an experiment all of whose outcomes are terminal is also called terminal.  The terminal outcomes do not have any emanating edges and correspond to the tree leaves. The data are therefore the frequencies observed at each leaf.   Even if the actual interventions that are applied at two distinct  nodes of the design are the same and have   the same set of outcomes, the probability distributions associated with them may be different. In that case, they will be considered as different experiments and thus different random variables.

The probabilities of outcomes of a single experiment will  be called the edge probabilities from now on. A subset of experiments with identical outcomes and outcome probabilities, or, equivalently, a subset of tree nodes with identical edge probabilities will be called a floret.  The terminal nodes are uniquely identified with the tree leaves and the path probabilities leading to terminal nodes will be called leaf probabilities. Given a tree node, the probability of the path leading to it equals the product of the edge probabilities along this path. The path probability of the root is assumed to be $1$. To illustrate, Figure \ref{fullTree} shows a tree design with three experiments $X_1$, $X_2$, $X_3$, and four leaves, marked by $p_1$, \dots, $p_4$. For instance, the node $X_2$  has two emanating edges with the edge probabilities $\theta_{21}$ and  $\theta_{22}$.   The path leading to the leaf whose probability is  $p_1$ consists of two edges parameterized by $\theta_{11}$ and $\theta_{21}$, respectively, so  the leaf probability is equal to $p_1 = \theta_{11} \theta_{21}$. 
Figure \ref{fullTreereduced}  shows a tree design with two experiments, $X_1$ and $X_2$,  and three leaves.  Because the path leading to the leaf $p_3$ is shorter than the paths leading to the other two leaves, the design is path dependent.

\begin{figure}[htbp]
  \centering
  \noindent
  \begin{minipage}[t]{0.45\linewidth}
    \includegraphics[width=\linewidth]{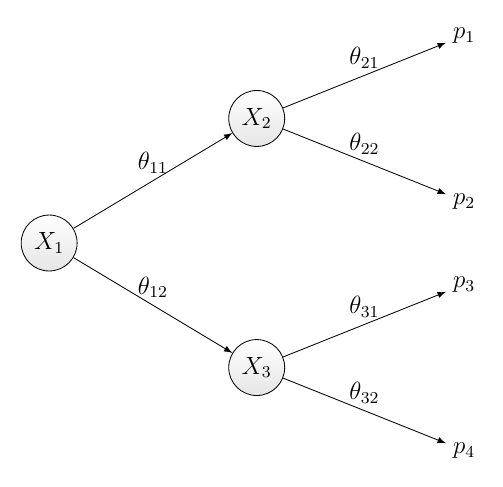}
    \caption{A three-experiment design of a complete tree structure.} 
    \label{fullTree}
  \end{minipage}%
  \hfill
  \begin{minipage}[t]{0.45\linewidth}
    \centering
    \includegraphics[width=\linewidth]{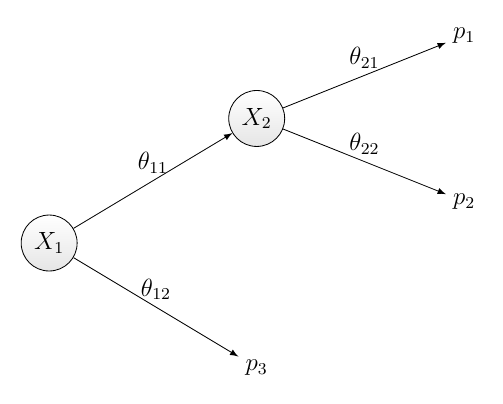}
    \caption{A two-experiment path-dependent design.} 
    \label{fullTreereduced}
  \end{minipage}
\end{figure}

 Denote the set of (non-terminal) nodes by ${X}_1, \dots, {X}_K$, with ${X}_1$ being the root. Assume that the probability distributions of ${X}_1, \dots, {X}_K$ have ranges $\mathcal{I}_1$,  \dots , $\mathcal{I}_K$ and are  parameterized by strictly positive edge probabilities $\boldsymbol \theta^{(1)} =  (\theta_{1l})_{l = 1}^{I_1}$, $\dots$, $\boldsymbol \theta^{(K)} = (\theta_{Kl})_{l = 1}^{I_K}$, respectively. Denote by $\mathcal{I}$ the set of leaf nodes and treat it as an ordered sequence, without loss of generality,  $\mathcal{I} = (1, \dots, I)$, where $I = |\mathcal{I}|$. 

This section studies the most restrictive hypothesis, stating that all $X_1$, \dots, $X_K$ are identically distributed, that is,  $\mathcal{I}_1  \equiv \cdots \equiv \mathcal{I}_K$, and
\begin{equation}\label{hypothesisOneClass}
\mathcal{H}:  \boldsymbol \theta^{(1)} =  \dots = \boldsymbol \theta^{(K)}. 
\end{equation}
In this case, the entire sequence of experiments forms a single floret.  Let $\boldsymbol \theta = (\theta_1, \dots, \theta_J)\tr$ denote the common floret parameter. Here, $ \boldsymbol \theta \in \Delta_J$, where $ \Delta_J$ is the open simplex in $J$ dimension.  For simplicity of presentation,  the set of positive probability distributions on leaves is identified with the open simplex $\Delta_I$ in $I$ dimensions.  Let $\mathbf{M} = (\mu_{ji})$ be a $J \times I$  matrix, such that for each $i = 1, \dots, I$, $j = 1, \dots, J$, the entry  $\mu_{ji}$ equals the number of times an edge parameter $\theta_j$ appears on the tree path from the root $X_1$ to the leaf $i$.  Given a tree $\mathcal{T}$  with leaf nodes $\mathcal{I}$, the hypothesis (\ref{hypothesisOneClass}) entails the set of constraints on the leaf parameters $\boldsymbol p$ that can be expressed as a multiplicative model with the design matrix $\mathbf{M}$.

\begin{definition} The one-floret tree model $\mathcal{M}$ generated by (\ref{hypothesisOneClass}) on a tree $\mathcal{T}$ with the nodes $X_1$, \dots, $X_K$ is the set of probability distributions on $\mathcal{I}$ that satisfy:
\begin{equation} \label{DefOneFloret}
\mathcal{M} = \left \{ \boldsymbol p \in \Delta_I: \,\boldsymbol p = \boldsymbol \theta^{\mathbf{M}\tr},\, \, \boldsymbol \theta = (\theta_1, \dots, \theta_J)\tr \in \Delta_J \right\}.
\end{equation} 
\end{definition}
The exponentiation is component-wise, that is,  \, $ \boldsymbol p = (p_i)_{i = 1}^I = ( \prod_{j = 1}^J \theta_j^{\mu_{ji}})_{i = 1}^I$.
After rewriting the definition in terms of the logarithms of $\log \theta$'s, one sees that a one-floret tree model (\ref{DefOneFloret}) is an exponential family of distributions \citep*[cf.][]{BrownBook, GeorgenTreesCurved2022}. 

\noindent Two examples of data structures leading to a one-floret tree model are described next.

\begin{example} \label{ExHWEq.1}
Consider a sequential experimental design with two treatments $T_1$ and $T_2$, each with a dichotomous outcome success/failure, and let $X_i$ denote the corresponding random variables. Suppose $T_1$ is applied first, and $T_2$ is applied second, irrespective of whether $T_1$ was successful or not. The data generating process can be visualized using the tree in Figure \ref{fullTree}. Although $X_2$ and $X_3$ are the same treatments but as they refer to the different subpopulations, the distributions of the outcomes may be different. In that case,
\begin{equation}\label{satur}
p_1 = \theta_{11} \theta_{21},  \quad p_2 = \theta_{11} \theta_{22},  \quad p_3 = \theta_{12} \theta_{31},  \quad p_4 = \theta_{12} \theta_{32}.
\end{equation}  

A possible hypothesis  assumes that the variables $X_2$ and $X_3$ have the same distribution but it is different from the one of $X_1$.  
Let $\theta_1$, $\theta_2 = 1- \theta_1$ denote the probabilities of success and failure for $X_1$, and $\theta_3$, $\theta_4 = 1 - \theta_3$ those for $X_2$ and $X_3$.   In this case, the leaf probabilities are equal to:
\begin{equation}\label{indep}
p_1 = \theta_1 \theta_3,  \quad p_2 = \theta_1 \theta_4,  \quad p_3 = \theta_2 \theta_3,  \quad p_4 = \theta_2\theta_4, 
\end{equation}  
which, in fact, is equivalent to the model of independence between $T_1$ and  $T_2$. 

A more restrictive hypothesis is that $X_1$, $X_2$, and $X_3$ are identically distributed, with the common probabilities of success $\theta_1$ and failure $\theta_2 = 1 - \theta_1$.  Then the leaf  probabilities satisfy: 
\begin{equation}\label{modelHWeq}
p_{1}= \theta_1^2, \,\,p_{2}=\theta_1\theta_2, \,\,p_{3}= \theta_2\theta_1, \,\,p_{4}= \theta_2^2. 
\end{equation} 
This hypothesis is equivalent to a tree model with one-floret $\{X_1, X_2, X_3\}$ and  the design matrix 
\begin{equation}\label{HWEmatrix}
\mathbf{M} = \left(\begin{array}{rrrr}   2 & 1 & 1 & 0 \\ 
                                                            0 & 1 & 1 & 2 \\
                                                             \end{array} \right). 
\end{equation} 
The matrix $\mathbf{M}$ has as many columns as the number of leaves and each row corresponds to an edge parameter.  A column of $\mathbf{M}$  indicates how many times each edge parameter appears on the path leading from the root to the corresponding leaf.
\qed
\end{example}

The next example illustrates a path-dependent data generation process.

\begin{example}
\label{ExCalvesDef}
\cite{Agresti2002} described a two-step study which investigated the existence of immunizing effect of a pneumonia infection on dairy calves.  First, all the calves were exposed to a pneumonia infection, and, second, the calves who got infected the first time, were exposed again.  After that the re-exposed calves were examined  whether or not they contracted the second infection.  The resulting data structure can be described using the following incomplete contingency table where the cell  ({No}, {Yes}) is empty, because  the calves that did not develop the first infection were not exposed to the second one.
$$
\begin{array}{ccc}
&\multicolumn{2}{c}{\mbox{Secondary Infection}}\\
\mbox{Primary Infection}& \mbox{Yes}& \mbox{No} \\
\mbox{ Yes }	& p_1 & p_2\\  
\mbox{ No } 	&   -   &  p_3\\ 
\end{array}
$$
\noindent Here $p_{1}$, $p_{2}$, $p_{3}$ are the probabilities of the corresponding {exposure/outcome} paths, so $p_{1} + p_{2} + p_{3} =1$. 
Notice that the table summary carries no information about how the data were generated.  To represent both the generation process and data structure itself  the  tree in Figure \ref{fullTreereduced} can be used. To express the hypothesis of no immunizing effect,  one assigns  the same edge  parameter $\theta_1 \in (0,1)$ as the probability of getting infected during the first and the second exposures, which results in the following leaf probabilities:
\begin{equation}\label{calP}
p_{1}=\theta_1^2, \,\,p_{2}=\theta_1\theta_2, \,\,p_{3}=\theta_2, \quad  \theta_1 + \theta_2 = 1,
\end{equation}
or in the matrix form,  $\boldsymbol p = \boldsymbol \theta^{\mathbf{M}\tr}$, where $\boldsymbol \theta = (\theta_1, \theta_2)\tr$ and the design matrix is 
\begin{equation} \label{calvesM}
\mathbf{M} = \left(\begin{array}{ccc} 2 & 1 & 0  \\ 
                                                       0 & 1 & 1  \\
                                                        \end{array} \right). \end{equation}  
The equations (\ref{calP}) define a one-floret model.

The hypothesis of no immunizing effect can also be specified by placing a restriction on the leaf probabilities, namely, that  $(p_1 + p_2)/ p_3 = p_1/p_2$, saying that the odds of success versus failure in the first and the second exposure are the same. This is an example of the very general odds ratios equalities that define the model and could be used to count degrees of freedom.  The two definitions of the model illustrate the duality of freedom and constraint specification of the model, as well.

\qed
\end{example}

Different experimental designs could also lead to the same data structure. If calves were not (actively) exposed to infection rather, were only observed whether they have acquired no or one or two infections during a period of time, 
then the data structure would be the same as in Example \ref{ExCalvesDef}.
However, in this case, the empty ({No},  {Yes}) cell would be implied by such a combination not being logically possible: a second infection without the first one cannot occur. In Example \ref{ExCalvesDef} the combination in same empty cell could exist, but was not observed by design. 

The relevance of distinguishing between logically impossible combinations and combinations not observed by design, is that models yielding positive estimates for the logically impossible combinations must not be used, while such a model may be relevant in the case of combinations that were not observed by design.
For details see \cite{KRD11}.

The models (\ref{modelHWeq}) and (\ref{calP}) were also described  as generalizations of relational models by  \cite{KRD11}, who also emphasized an intrinsic difference between the two models, namely, the presence of the overall effect (OE) in the exponential family corresponding to the former model  and its absence in the exponential family corresponding to the latter. The presence of the OE is equivalent to the existence of a model parameterization in which all $p_i$'s can have a common parameter,  also serving as normalization constant.  While the presence of the OE is rather a conventional assumption,  the absence of the OE  is usually characteristic to the data generation process.  The models without an OE become relevant when the sample space is an incomplete Cartesian product and occur naturally for incomplete sequential designs \citep[cf.][]{KRoveff}.  The next proposition states a sufficient condition for a one-floret model to have the OE.

\begin{proposition} A one-floret tree model entailed by a complete tree  has the overall effect.  
\end{proposition}
\begin{proof}

\noindent Because in a complete tree  all root-to-leaf paths are of the same length,   the sums across the columns of the design matrix $\mathbf{M}$  are also the same, which implies that the row of all 1's,  $\boldsymbol 1 \tr = (1, \dots, 1)$,  is in the row space of $\mathbf{M}$. Therefore,  the model can be reparameterized to have a common parameter in all $p_i$'s, so $\mathcal{M}$ is a model with  the OE.  
\end{proof}

For instance,  the tree in Figure \ref{fullTree} is complete because all root-to-leaf paths have the same length, respectively, the column sums of the design matrix (\ref{HWEmatrix}) are the same,  and thus  (\ref{modelHWeq}) is a model with the OE.  The tree in Figure \ref{fullTreereduced} is incomplete and it can be checked that the row space of the design matrix (\ref{calvesM}) does not contain the row of 1's $(1,1,1)$, thus the corresponding model (\ref{calP})  does not have the OE.  In the latter case,  including an additional parameter, as the OE, to the model would lead to a saturated model without the desired tree structure   \citep{KRD11}.

The maximum-likelihood estimation under a one-floret tree model is discussed next. Let $\boldsymbol y$  be the observed data obtained from a sample of size $N$ using a generating procedure of a tree structure with nodes ${X}_1, \dots, {X}_K$. As follows from  Theorem   4.6 of \cite{TRbook2018}, $\boldsymbol y$ is a realization of a random variable $\boldsymbol Y$ which has a multinomial distribution $Mult(N, \boldsymbol p)$ with parameters $N$ and $\boldsymbol p$, where $\boldsymbol p$ is the vector of leaf probabilities. 
The kernel of the multinomial log-likelihood function under a one-floret model $\mathcal{M}$ is equal to:
\begin{equation}\label{LLoneF}
{L}(\boldsymbol p(\boldsymbol{\theta}), \boldsymbol y) =  {M}_1\boldsymbol y \log   \theta_1 + \dots {M}_J\boldsymbol y  \log   \theta_J  = (\mathbf{M}\boldsymbol y)\tr  \log \boldsymbol \theta,
\end{equation} 
where  $M_1,  \dots,  M_J$ denote the rows of the design matrix $\mathbf{M}$ and  $\log  \boldsymbol \theta = (\theta_1 , \dots, \theta_J)\tr$. The MLE of $\boldsymbol{\theta}$, if exists, is the unique maximum of the Lagrange function:
\begin{equation}\label{LLoneF1}
\mathcal{{L}}(\boldsymbol p(\boldsymbol{\theta}), \boldsymbol \alpha, \boldsymbol y) = (\mathbf{M}\boldsymbol y)\tr \log  \boldsymbol \theta  -  \alpha(\boldsymbol 1' \boldsymbol \theta - 1),
\end{equation} 
where $ \alpha$ stands for the Lagrange multiplier. The theory of exponential families imply that  ${M}_1\boldsymbol y,$  \dots,  ${M}_J\boldsymbol y$ are the sufficient statistics of the MLE under the model $\mathcal{M}$ \citep[cf.][]{BrownBook}.   The sum, $\mathcal{S}(\boldsymbol y) = {M}_1\boldsymbol y \dots +{M}_J\boldsymbol y = \boldsymbol 1\tr\mathbf{M}\boldsymbol y$,  is the total number of experiments performed to obtain the data.  In the sequel,  $\mathcal{S}(\boldsymbol y)$ will be called \textit{exposure size}.

\begin{theorem}\label{theoremMatrixOneFloret}
Let $\mathcal{M}$ be a one-floret tree model (\ref{DefOneFloret}) with a design matrix $\mathbf{M}$.   Assume that  the MLEs of the edge probabilities $\hat{\boldsymbol \theta}$ and  leaf probabilities $\hat{\boldsymbol p}$  given $\boldsymbol y$  exist. Then, the following statements hold:
\begin{enumerate}[(i)]
\item 
The MLEs are equal to, respectively, 
\begin{equation}\label{MLEOneclosed1}
\hat{\boldsymbol \theta}=   \frac{1}{\mathcal{S}({\boldsymbol{y}})}  \mathbf{M}\boldsymbol y,  \quad  \hat{\boldsymbol p} = \hat{\boldsymbol \theta}^{\mathbf{M}\tr}.
\end{equation}
\item For each $j = 1 \dots, J$ and $\hat{\boldsymbol y} = N\hat{\boldsymbol p}$,
\begin{equation}\label{adjFoneFloretJ}
\frac{{M}_j\hat{\boldsymbol y}} {\mathcal{S}(\hat{\boldsymbol y})}=  \frac{{M}_j {\boldsymbol y}}{\mathcal{S}({\boldsymbol{y}})}.
\end{equation}
\item Further,
\begin{equation}\label{adjFoneFloretRiskJ}
\frac{{M}_1\hat{\boldsymbol y}} {{M}_1 {\boldsymbol y}}= \dots = \frac{{M}_J\hat{\boldsymbol y}} {{M}_J{\boldsymbol y}} = \frac{\mathcal{S}(\hat{\boldsymbol y})}{\mathcal{S}({\boldsymbol y})}.
\end{equation}
\item Finally,   ${\mathcal{S}(\hat{\boldsymbol y})} = {\mathcal{S}({\boldsymbol y})}$ if and only if the row $\boldsymbol 1\tr = (1, \dots, 1)$ is in the row space of $\mathbf{M}$. 
\end{enumerate}
\end{theorem}

\begin{proof}
\begin{enumerate}[(i)]
\item The differentiation of the Lagrangian (\ref{LLoneF1}) with respect to $\boldsymbol \theta$ and $\alpha$ yields
\begin{align}\label{likeOneEqf_M}
&{\partial \mathcal{{L}}}/{\partial \boldsymbol \theta} = diag({\boldsymbol \theta}^{-1})\mathbf{M}\boldsymbol y -
\alpha \boldsymbol 1 = \boldsymbol 0, \\
&{\partial \mathcal{{L}}}/{\partial  \alpha} = \boldsymbol 1'\boldsymbol \theta - 1 = 0.\nonumber
\end{align}
Left-multiplying the first equation of (\ref{likeOneEqf_M}) by $diag({\boldsymbol \theta})$ yields that $\mathbf{M}\boldsymbol y = \alpha \boldsymbol \theta$, and thus, from the second equation  of (\ref{likeOneEqf_M}), $\hat{\alpha} =  \boldsymbol 1 \tr \mathbf{M}\boldsymbol y = \mathcal{S}({\boldsymbol{y}})$. Therefore,  
\begin{equation*}
\hat{\boldsymbol \theta}=   \frac{1}{\mathcal{S}({\boldsymbol{y}})}{\mathbf{M}\boldsymbol y},   
\end{equation*}
and, from (\ref{DefOneFloret}), $\hat{\boldsymbol p} = \hat{\boldsymbol \theta}^{\mathbf{M}\tr}$. 
\item[(ii, iii)] To prove, one can use, for example, the generalized mean-value theorem of \cite{KRD11}. Because the model $\mathcal{M}$ is an exponential family that satisfies the conditions of this Theorem 3.3,  there exists a unique $\gamma > 0$, such that,
\begin{equation}\label{subsetSumsOneFloret}
{M}_1 \hat{\boldsymbol p}= \frac{\gamma}{N}M_1{\boldsymbol y}, \,\, \dots,  \quad {M}_J\hat{\boldsymbol p}= \frac{\gamma}{N}M_J{\boldsymbol y}.
\end{equation}
Thus, 
$$
{\mathcal{S}(\hat{\boldsymbol y})} = \boldsymbol 1\tr \mathbf{M} \hat{\boldsymbol y} = N\sum_{j = 1}^J {M}_j \hat{\boldsymbol p} =  \gamma \sum_{j = 1}^J {M}_j {\boldsymbol y}  =   \gamma \boldsymbol 1\tr \mathbf{M} {\boldsymbol y} = \gamma \mathcal{S}(\boldsymbol y),$$
from which 
\begin{equation}\label{gammaOneFloret}
{\gamma} = 
\frac{\mathcal{S}(\hat{\boldsymbol y})}{\mathcal{S}({\boldsymbol{y}})}.
\end{equation}
The statements (\ref{adjFoneFloretJ}) and (\ref{adjFoneFloretRiskJ}) are now immediate.
\item[(iv)]  Theorem 3.3 of \cite{KRD11} applies. The condition that $\boldsymbol 1\tr = (1, \dots, 1) \in rowspace (\mathbf{M})$ is necessary and sufficient for $\gamma = 1$, that is,   for ${\mathcal{S}(\hat{\boldsymbol y})} = {\mathcal{S}({\boldsymbol{y}})}$. 

\end{enumerate}
\end{proof}

The value of  $\gamma$ found  in (\ref{gammaOneFloret}) is the ratio between the expected and observed exposure size, and in the sequel will be referred to as \textit{exposure ratio}. Theorem implies that, for a model $\mathcal{M}$, the sufficient statistics of the MLE, $\mathbf{M} \hat{\boldsymbol y} $, are, in general, not equal but only proportional to  the observed sufficient statistics, $\mathbf{M} {\boldsymbol y}$, with the common coefficient of proportionality, $\gamma$, equal to the exposure ratio. 

The MLE computation  is demonstrated using Examples  \ref{ExHWEq.1} and \ref{ExCalvesDef}.  

\begin{repexample} {ExHWEq.1}\textbf{(revisited)}
Given an observed frequency distribution $\boldsymbol y = ( y_{1},  y_{2},  y_{3}, y_{4}) \tr \sim Mult(N, (p_{1}, p_{2}, p_{3}, p_{4}))$,  the kernel of the log-likelihood function under the model (\ref{modelHWeq}) is equal to 
$L(\boldsymbol \theta, \boldsymbol y) = (2y_{1} +  y_{2} + y_{3})\log \theta_1 + (y_{2} + y_{3} +  2y_{4})\log\theta_2$.  The statistics $(2y_{1} +  y_{2} + y_{3})$ and $(y_{2} + y_{3} +  2y_{4})$ are the sufficient statistics of the model. It is straightforward to show that it is maximized by
$$\hat{\theta}_1  = \frac{2y_{1} +  y_{2} + y_{3}}{2N}, \quad \hat{\theta}_2  = \frac{y_{2} +  y_{3} + 2y_{4}}{2N}.$$
The MLE of $\boldsymbol p$ is equal to:
\small{
\begin{align*}
\hat{\boldsymbol p} &=\left(\hat{\theta}_1^2,\, \hat{\theta}_1\hat{\theta}_2,\, \hat{\theta}_1\hat{\theta}_2,\, \hat{\theta}_2^2\right)\tr \\
&= \left (\frac{(2y_{1} + y_{2} +  y_{3})^2}{4N^2}, \frac{(2y_{1} + y_{2} +  y_{3})(y_{2} + y_{3} +  2y_{4})}{4N^2}, \frac{(2y_{1} + y_{2} +  y_{3}) (y_{2} + y_{3} +  2y_{4})}{4N^2}, \right.\\
&\left. \frac{(y_{2} + y_{3} +  2y_{4})^2}{4N^2}\right)\tr.
\end{align*}
}
All column sums of the design matrix  $\mathbf{M}$ are equal to $2$, and therefore, $\boldsymbol 1\tr \in rowspace (\mathbf{M})$.  As follows from Theorem \ref{theoremMatrixOneFloret}(iv), the \text{exposure ratio}  $\gamma = 1$, which can be also verified directly by substitution.
\qed
\end{repexample} 

\vspace{3mm}

\begin{repexample} {ExCalvesDef}\textbf{(revisited)}

Given $\boldsymbol y = ( y_{1},  y_{2},  y_{3}) \tr \sim Mult(N, (p_{1}, p_{2}, p_{3}))$, the kernel of log-likelihood function is equal to 
$L(\boldsymbol \theta, \boldsymbol y) = (2y_{1} +  y_{2})\log \theta_1 + (y_{2} +  y_{3})\log\theta_2$. 
The statistics $(2y_{1} +  y_{2})$ and $(y_{2} +  y_{3})$ are the sufficient statistics of the model, and the MLEs of $\theta_1$ and  $\theta_2$ are equal to:
$$\hat{\theta}_1  = \frac{2y_{1} +  y_{2}}{2y_{1} + 2 y_{2} + y_{3}}, \quad \hat{\theta}_2  = \frac{y_{2} +  y_{3}}{2y_{1} + 2 y_{2} + y_{3}}.$$ 
The result is intuitive,  because the denominator $2y_{1} + 2 y_{2} + y_{3}$ is equal to the total number of times the calves were exposed to the infection, and the numerator for  $\hat{\theta}_1$ is the number of calves who acquired the infection at least once, while the numerator for  $\hat{\theta}_2$ is the number of calves who did not acquire any infection. 
The MLEs of the leaf probabilities are equal to:
\begin{align*}
\hat{p}_{1} &= \hat{\theta}_1^2 = \frac{(2y_{1} +  y_{2})^2}{(2y_{1} + 2 y_{2} + y_{3})^2}, \quad \hat{p}_{2} = \hat{\theta}_1 \hat{\theta}_2 = \frac{(2y_{1} +  y_{2})(y_{2} +  y_{3})}{(2y_{1} + 2 y_{2} + y_{3})^2}, \\
\hat{p}_{3} &= \hat{\theta}_2 = \frac{y_{2} +  y_{3}}{(2y_{1} + 2 y_{2} + y_{3})}.
\end{align*}

Notice that the row space of the design matrix $\mathbf{M}$ in (\ref{calvesM})  does not contain the row $(1,1,1)$.   The estimated values of sufficient statistics are proportional,  but, in general,  not equal, to those observed: 
$$N(2\hat{p}_{1} +  \hat{p}_{2}) = \gamma (2y_{1} +  y_{2}), \quad N(\hat{p}_{2} + \hat{p}_{3})  = \gamma (y_{2} +  y_{3}),$$
and the exposure ratio $\gamma$ is equal to: 
\begin{align*} 
{\gamma} &= \frac{N(2\hat{p}_{1} +  2\hat{p}_{2} + \hat{p}_{3})}{2y_{1} + 2 y_{2} + y_{3}}.  
\end{align*}
The numerator in ${\gamma}$ can be written as:  
$N(2\hat{p}_{1} +  2\hat{p}_{2} + \hat{p}_{3})$ $= {N(2 \hat{\theta}_1^2 +  2\hat{\theta}_1 \hat{\theta}_2 + \hat{\theta}_2)}$  $=  {N(\hat{\theta}_1 + 1)},$ and is indeed the expected number of exposures.
 \qed 
\end{repexample}

\vspace{3mm}

The asymptotic distributions of the MLE of the model parameters and leaf probabilities are derived next. 
Let $\mathcal{M}$ be a   one-floret model  (\ref{DefOneFloret}) parameterized by $\boldsymbol \theta = (\theta_1,  \dots, \theta_{J})\tr$.   By the model definition,  the sum-to-one constraint, $\boldsymbol 1\tr\boldsymbol \theta = 1$, is the only restriction on model parameters and, therefore,  by setting $\theta_J =  1 - \theta_1 - \dots - \theta_{J-1}$,  the model can be expressed  in terms of $J-1$ non-redundant parameters, $\tilde{\boldsymbol \theta} \coloneqq(\theta_1,  \dots, \theta_{J-1})\tr$. The parameter space is therefore $\boldsymbol \Omega = \{\tilde{\boldsymbol \theta} \in (0,1)^{J-1}: \,\,\boldsymbol 1\tr\tilde{\boldsymbol \theta} < 1 \}$. 
Let $\boldsymbol \theta_0 = (\tilde{\boldsymbol \theta}_0, {\theta}_{J0})\tr$ denote the corresponding true values. The simplify the notation,   the value  of the matrix of the first derivatives of $ \boldsymbol p$ with respect to  $\tilde{\boldsymbol \theta}$, $(\partial \boldsymbol p /\, \partial \tilde{\boldsymbol \theta})$,  at $\boldsymbol \theta_0$ will be denoted by  $(\partial \boldsymbol p / \, \partial \boldsymbol \theta_0)$. Finally, let $\mathbf{A} = diag(\boldsymbol p_0)^{-1/2}(\partial \boldsymbol p /\, \partial \boldsymbol \theta_0)$.

\begin{theorem}\label{theoremAgrestiAsymptotic}
Let $\mathcal{M}$ be a one-floret model and assume that the true value $\boldsymbol p_0  = \boldsymbol p(\boldsymbol \theta_0) \in \mathcal{M}$.  Given data $\boldsymbol y \sim Mult(N, \boldsymbol p_0)$, if the MLE $\hat{\tilde{\boldsymbol \theta}}$ of $\tilde{\boldsymbol \theta}_0$ and $\hat{\boldsymbol p} =  \boldsymbol p (\hat{\boldsymbol \theta})$ of $\boldsymbol p_0$ exist,  the asymptotic distributions of $\hat{\tilde{\boldsymbol \theta}}$ and $\hat{\boldsymbol p}$ are, respectively: 
$$\sqrt{N}(\hat{\tilde{\boldsymbol \theta}} - \tilde{\boldsymbol \theta}_0) \overset{d}{\rightarrow} \mathcal{N}(\boldsymbol 0, \boldsymbol{\Phi}_{{\tilde{\boldsymbol{\theta}}}}),$$
$$\sqrt{N}(\hat{\boldsymbol p} - \boldsymbol p_0) \overset{d}{\rightarrow} \mathcal{N}(\boldsymbol 0, \boldsymbol{\Phi}_{{\boldsymbol{p}}} ),$$
\noindent where $\overset{d}{\rightarrow}$ denotes convergence in distribution, and the corresponding asymptotic covariance matrices  are equal to:  
\begin{equation}\label{AsCovProbOneFloret}
\boldsymbol{\Phi}_{{\tilde{\boldsymbol{\theta}}}} =(\mathbf{A}\tr \mathbf{A})^{-1},\quad \quad \boldsymbol{\Phi}_{{\boldsymbol{p}}} =\frac{\partial \boldsymbol p}{\partial \boldsymbol \theta_0}(\mathbf{A}\tr \mathbf{A})^{-1}\frac{\partial \boldsymbol p\tr}{\partial \boldsymbol \theta_0}.
\end{equation}
\end{theorem}

The statements follow from \cite{Birch64proof}, see also Theorem 14.8-4 in \cite{BFH}, among others, and the proof is omitted.  For a one-floret model a simpler expression for $\boldsymbol{\Phi}_{\tilde{\boldsymbol{\theta}}}$  can be obtained. 

\begin{theorem}\label{theoremMatrixProbOneFloret}
For a one-floret model $\mathcal{M}$ the following holds:
\begin{enumerate}[(i)]
\item 
\begin{equation}\label{OneFloretMxA}
\mathbf{A}\tr \mathbf{A} = \mathcal{S}(\boldsymbol p(\boldsymbol{\theta}_0))( diag(\tilde{\boldsymbol \theta}_0^{-1})  + \boldsymbol 1 \boldsymbol 1\tr  \theta_{J0}^{-1}),
\end{equation}
where $\mathcal{S}(\boldsymbol p(\boldsymbol{\theta}_0)) =  \boldsymbol{1} \tr  \mathbf{M} \boldsymbol{p}(\boldsymbol{\theta}_0)$ is the exposure total evaluated at $\boldsymbol p(\boldsymbol \theta_0)$.
\item The asymptotic covariance matrix of $\sqrt{N}\hat{\tilde{\boldsymbol \theta}}$ is equal to:  
\begin{equation}\label{AsCovMultiOneFloret}
\boldsymbol{\Phi}_{\tilde{\boldsymbol{\theta}}} =\frac{1}{ {\mathcal{S}(p(\boldsymbol{\theta}_0))}} (diag(\tilde{\boldsymbol \theta}_0)  - \boldsymbol 1  \boldsymbol 1\tr  \theta_{J0}). \end{equation}
\end{enumerate}
\end{theorem}

\begin{proof}
The proof of (\textit{i}) is given in  the Appendix. The statement (\textit{ii}) follows from (\ref{AsCovProbOneFloret}) and  from the equation 
$$\left(diag(\tilde{\boldsymbol{\theta}}_0^{-1}) + \boldsymbol 1 \boldsymbol 1\tr \theta_{J0}^{-1}\right)^{-1} = 
diag(\tilde{\boldsymbol \theta}_0)  - \boldsymbol 1 \boldsymbol 1\tr  \theta_{J0}.
$$
\end{proof}

\vspace{3mm}
Notice that the matrix $\boldsymbol{\Sigma}_{\tilde{\boldsymbol{\theta}}} = diag(\tilde{\boldsymbol{\theta}}_0) - \boldsymbol 1  \boldsymbol 1\tr \theta_{J0}$ that appears in (\ref{AsCovMultiOneFloret}) is the covariance matrix of ${\tilde{\boldsymbol{\theta}}}_0$.
Then,  $\boldsymbol{\Phi}_{\tilde{\boldsymbol{\theta}}}$ can be rewritten as a product:
$$\boldsymbol{\Phi}_{\tilde{\boldsymbol{\theta}}} = \frac{N}{ {\mathcal{S}(Np(\boldsymbol{\theta}_0))}}  \boldsymbol{\Sigma}_{\tilde{\boldsymbol{\theta}}}.$$ 
Here, the factor $\frac{N}{ {\mathcal{S}(Np(\boldsymbol{\theta}_0))}} $ is the reciprocal of the ratio between the true value of the exposure size, ${\mathcal{S}(Np(\boldsymbol{\theta}_0))}$, and the sample size, $N$.  In the remainder of the paper this ratio will be referred to as \textit{exposure rate}.

\begin{repexample} {ExHWEq.1}\textbf{(revisited)}
The subscript $0$ indicating a true parameter value is omitted further for brevity. After expressing the model  (\ref{modelHWeq}) in terms of $\theta_1$ only, 
$$p_{1}= \theta_1^2, \,\,p_{2}=\theta_1(1-\theta_1), \,\,p_{3}=(1-\theta_1)\theta_1, \,\,p_{4}=(1 - \theta_1)^2,$$
one obtains
\begin{equation} \label{AMTwoFlorets}
\mathbf{A} = diag(\boldsymbol{p})^{-1/2}\frac{\partial \boldsymbol p}{\partial \theta_1} = \left(\begin{array}{c} 2\theta_1(\theta_1)^{-1}\\
                                 (1- 2\theta_1)(\theta_1(1-\theta_1))^{-1/2}\\
                                 (1- 2\theta_1)(\theta_1(1-\theta_1))^{-1/2}\\
                                 -2(1-\theta_1)(1-\theta_1)^{-1}\end{array}\right),  \quad \mathbf{A}\tr \mathbf{A}  = \frac{2}{\theta_1(1-\theta_1)}.
                                 \end{equation}
Thus $\hat{\boldsymbol{\Phi}}_{\tilde{\theta}_1} = {\hat{\theta}_{1}(1-\hat{\theta}_{1})}/{2}$. Both the true exposure rate and its expected value are equal to $2$, which is quite logical because the floret treatment is applied  to each subject two times. \qed
\end{repexample} 


\begin{repexample} {ExCalvesDef}\textbf{(revisited)}
It is straightforward to show that
$\boldsymbol{\Phi}_{\tilde{\theta}_1} $ \,$= {\theta_1(1-\theta_1)}/{(\theta _1+ 1)}.$ The exposure rate is equal to $\theta_1 + 1$, meaning that,  given a sample of size $N$,  the floret treatment is expected to be applied $N(\theta_1 + 1)$ times. In this case, the expected exposure rate is unknown in advance and depends on the observed data.
 \qed 
\end{repexample}

\section{Multi-floret sequential designs}\label{SectionMultiple}

Consider an experimental design of tree structure as described at the beginning of Section \ref{sectionOneTrt}.   The tree nodes $X_1, \dots, X_K$ correspond to the experiments or interventions whose outcomes are random variables parameterized by strictly positive $\boldsymbol \theta^{(1)} $, $\dots$, $\boldsymbol \theta^{(K)}$, where $\boldsymbol \theta^{(k)} =  (\theta_{kl})_{l = 1}^{I_k}$, $k = 1, \dots, K$. This section focuses on testing a hypothesis stating that some of the responses to some interventions are identically distributed, that is, there exists a partition $\mathcal{F} = \{S_1, \dots, S_F \}$ of the nodes $\{X_1, \dots, X_K \}$ into subsets $S_f$, such that  for any $s, t \in \{1, \dots, K\}, \, s\neq t$, and $f = 1, \dots, F$,

\begin{align*}
&X_s, X_t \in S_f \iff \boldsymbol \theta^{(s)} = \boldsymbol \theta^{(t)}. 
\end{align*}
Recall that the subsets of nodes with the same edge probabilities are called \textit{florets}. Denote by $\boldsymbol \theta_f$ the common edge parameters of the distributions in $S_f$ and then $\boldsymbol \theta = (\boldsymbol \theta_f)_{S_f \in \mathcal{F}}$. Let $J_f =|\boldsymbol \theta_f|$ and $J =|\boldsymbol \theta|$ stand for the lengths of $\boldsymbol \theta_f$ and $\boldsymbol \theta$, respectively. Further, from the root $X_1$ there leads a unique path to a leaf $i$, where $i = 1, \dots, I$.  Let $\mu_{fji}$ denote the number of times the edge parameter $\theta_{fj}$ appears on this path.  Then $\mathbf{M}_f = (\mu_{fji})$ is $J_f \times I$  matrix. Finally, let $\mathbf{M}$ be the $J \times I$ block matrix with blocks $\mathbf{M}_f$. In the sequel, to simplify the notation, a floret $S_f$ is referred to by its index, $f$.  A floret $f$ is called complete if the row space of $\mathbf{M}_f$ contains a row of $1$'s.

\begin{definition} The multi-floret tree model $\mathcal{M}$ generated by a set of florets $\mathcal{F} = \{S_1, \dots, S_F \}$ is the set of positive distributions on $\mathcal{I}$ that satisfy: 
\begin{equation} \label{PMmatr}
\mathcal{M} = \left \{ \boldsymbol p \in \mathcal{P}:  \,\,{\boldsymbol p} = \prod_{f \in \mathcal{F}} ({\boldsymbol \theta}_f)^{\mathbf{M}_f\tr}, \,\, \boldsymbol \theta = (\boldsymbol \theta_f)_{f \in \mathcal{F}} \right\}.
\end{equation} 
\end{definition}

 Let $\boldsymbol y$  be the observed data obtained from a sample of size $N$ using a generating procedure of a tree structure with nodes $\boldsymbol{X}_1, \dots, \boldsymbol{X}_K$.   As noted in the previous section,   $\boldsymbol y$ is a realization of a random variable $\boldsymbol Y$ which has a multinomial distribution $Mult(N, \boldsymbol p)$ where $\boldsymbol p$ is the vector of leaf probabilities.   Notice that that under the model (\ref{PMmatr}), the kernel of the multinomial log-likelihood function ${L}(\boldsymbol p(\boldsymbol{\theta}), \boldsymbol y) $ can  be partitioned into floret-specific terms, namely:
\begin{equation}\label{LLmultiF}
{L}(\boldsymbol p(\boldsymbol{\theta}), \boldsymbol y) = \boldsymbol y \tr \log  \boldsymbol p(\boldsymbol{\theta}) = \sum_{f \in \mathcal{F} }(\mathbf{M}_f\boldsymbol y)\tr  \log \boldsymbol \theta_f =  \sum_{f \in \mathcal{F} } {L}_f(\boldsymbol p(\boldsymbol{\theta}), \boldsymbol y).
\end{equation} 
(The logarithm is applied component-wise.)  Here, each  ${L}_f(\boldsymbol p(\boldsymbol{\theta}), \boldsymbol y)  = \mathbf{M}_f\boldsymbol y \tr \log  \boldsymbol \theta_f$ is the kernel of the floret-specific multinomial log-likelihood function.
Because the floret parameters $\boldsymbol \theta_f$ mutually variation independent, Theorem 4.6 of \cite{TRbook2018} implies that the MLE   of $\boldsymbol{\theta} = (\boldsymbol \theta_f)_{f \in \mathcal{F}}$ can be found by separately maximizing the floret-specific log-likelihood functions.  
Consequently, the MLE $\hat{\boldsymbol{\theta}}$ and its properties can be obtained by extending Theorems \ref{theoremMatrixOneFloret}  to the multi-floret case, as shown next. In the sequel, $\mathcal{S}_f({\boldsymbol{y}}) = \boldsymbol 1_f\tr \mathbf{M}_f \boldsymbol y$ is the exposure size in the floret $f$. Here $\boldsymbol 1_f$ is the vector of $1's$ of length $J_f$.

\begin{theorem}\label{theoremMatrixMultiFloret}
Let $\mathcal{M}$ be a multi-floret tree model (\ref{PMmatr}).   Assume that  the MLEs $\hat{\boldsymbol \theta}$ and $\hat{\boldsymbol p}$   given $\boldsymbol y$  exist.  Then, the following holds:
\begin{enumerate}[(i)]
\item The MLE of the edge probabilities and leaf probabilities are equal to, respectively:
\begin{equation}\label{MLEclosed1}
\hat{\boldsymbol \theta}_{f} =  \frac{1}{\mathcal{S}_f({\boldsymbol{y}})} \mathbf{M}_f \boldsymbol y, \, f \in \mathcal{F},  \quad \hat{\boldsymbol p} = \prod_{f \in \mathcal{F}} (\hat{\boldsymbol \theta}_f)^{\mathbf{M}_f\tr}.
\end{equation}
\item For each $f = 1 \dots, F$ and $\hat{\boldsymbol y} = N\hat{\boldsymbol p}$,
\begin{equation}\label{adjFmultiFloretJ}
 \frac{1}{\mathcal{S}_f(\hat{\boldsymbol{y}})}  \mathbf{M}_f \hat{\boldsymbol y} =  \frac{1}{\mathcal{S}_f({\boldsymbol{y}})}  \mathbf{M}_f {\boldsymbol y} .
\end{equation}
\item For each $f \in \mathcal{F}$, the floret exposure ratio,
\begin{equation}\label{adjMultiFloretGamma} {\gamma}_f = 
\frac{\mathcal{S}_f(\hat{\boldsymbol y})}{\mathcal{S}_f({\boldsymbol{y}})},
\end{equation}
is equal to $1$  if and only if the row space of $\mathbf{M}_f$ contains a row of $1$'s. 
\end{enumerate}
\end{theorem}

\begin{proof}

For each  $f \in \mathcal{F}$, if the MLE  $\hat{\boldsymbol \theta}_f$ exists, then by Theorem \ref{theoremMatrixOneFloret},  it is equal to
\begin{equation*}
\hat{\boldsymbol \theta}_{f} =   \frac{1}{\mathcal{S}_f({\boldsymbol{y}})} \mathbf{M}_f \boldsymbol y, 
\end{equation*}
and therefore, $\hat{\boldsymbol p} = \prod_{f \in \mathcal{F}} (\hat{\boldsymbol \theta}_f)^{\mathbf{M}_f\tr}$.  The statement (ii) follows from (ii) of Theorem \ref{theoremMatrixOneFloret} and, finally, (iii) follows from the same theorem, part (iv),  applied floret-wise.

\end{proof}

\vspace{5mm}

Notice that if the row space of a floret submatrix $\mathbf{M}_f$ contains a row of $1$'s then the corresponding floret has the overall effect (OE), which happens, for example, when the floret is complete.  A multi-floret model can include both florets with and without the OE.

The asymptotic results for the MLE of $\boldsymbol{\theta}$  and of $\boldsymbol{p}$ in a multi-floret case are presented next.   Firstly, because within each floret $\theta_{f,J_f} =  1 - \theta_{f1} - \dots - \theta_{f,J_f-1}$, the multi-floret model can be expressed  in terms of $J-f$ non-redundant parameters, $\tilde{\boldsymbol \theta} = (\tilde{\boldsymbol \theta}_f)_{f \in \mathcal{F}}$, where $\tilde{\boldsymbol \theta}_f \coloneqq(\theta_{f1},  \dots, \theta_{f,J_f-1})\tr$.   For each $f = 1, \dots, F$,  $\boldsymbol{\theta}_{f0} = (\tilde{\boldsymbol \theta}_{f0}, {\theta}_{f,J_f,0})\tr$ denote the corresponding true values of the floret-specific components.  Similarly to the previous section, let $(\partial \boldsymbol p / \, \partial \boldsymbol \theta_0)$ denote the value of the matrix of first derivatives $(\partial \boldsymbol p /\, \partial \tilde{\boldsymbol \theta})$ of $\boldsymbol p$ with respect  to $\tilde{\boldsymbol \theta}$ at $\boldsymbol \theta_0$, and the matrix $\mathbf{A} = diag(\boldsymbol p_0)^{-1/2}(\partial \boldsymbol p /\, \partial \boldsymbol \theta_0)$.

\begin{theorem}\label{theoremAgrestiAsymptoticMulti}
Let $\mathcal{M}$ be a multi-floret model with a set of florets $\mathcal{F}$ and assume that the true parameter value $\boldsymbol p_0  = \boldsymbol p(\boldsymbol \theta_0) \in \mathcal{M}$.  Assume that given data $\boldsymbol y \sim Mult(N, \boldsymbol p_0)$ the MLE $\hat{\tilde{\boldsymbol \theta}}$ and $\hat{\boldsymbol p} =  \boldsymbol p (\hat{\boldsymbol \theta})$ exist.  As $N \to \infty$,  the  asymptotic distributions of $\hat{\tilde{\boldsymbol \theta}}$ and $\hat{\boldsymbol p}$ are, respectively: 
\begin{align}\label{systemPnewMatrix_estimate_M1}  
&\sqrt{N}(\hat{\tilde{\boldsymbol \theta}} - \tilde{\boldsymbol \theta}_0) \overset{d}{\rightarrow} \mathcal{N}(\boldsymbol 0, \boldsymbol{\Phi}_{{\tilde{\boldsymbol{\theta}}}}); \\
&\sqrt{N}(\hat{\boldsymbol p} - \boldsymbol p_0) \overset{d}{\rightarrow} \mathcal{N}(\boldsymbol 0, \boldsymbol{\Phi}_{{{\boldsymbol{p}}}}), \nonumber
\end{align}
\noindent where $\overset{d}{\rightarrow}$ denotes convergence in distribution, and  
\begin{equation}\label{AsCovProbMultiFloret}
\boldsymbol{\Phi}_{{\tilde{\boldsymbol{\theta}}}} =(\mathbf{A}\tr \mathbf{A})^{-1}, \quad \boldsymbol{\Phi}_{{\boldsymbol{p}}} =\frac{\partial \boldsymbol p}{\partial \boldsymbol \theta_0}(\mathbf{A}\tr \mathbf{A})^{-1}\frac{\partial \boldsymbol p\tr}{\partial \boldsymbol \theta_0}.
\end{equation}
\end{theorem}

The statements (\ref{systemPnewMatrix_estimate_M1}) and (\ref{AsCovProbMultiFloret}) follow from the asymptotic theory for general parametric models for discrete data \citep{Birch64proof, BFH}. The proof is omitted.  As in the one-floret case,  a simpler expression for $\boldsymbol{\Phi}_{\hat{\boldsymbol{\theta}}}$  can be derived. 

\begin{theorem}\label{theoremMatrixProbMultiFloret}
For a multi-floret model $\mathcal{M}$,
\begin{enumerate}[(i)]
\item The matrix $\mathbf{A}\tr \mathbf{A} $ is block-diagonal, constructed from floret-specific blocks:
\begin{equation*}
\mathbf{A}\tr \mathbf{A} = diag\left\{\mathcal{S}_f(\boldsymbol p(\boldsymbol{\theta}_0))( diag(\tilde{\boldsymbol \theta}_{f0}^{-1})  + \boldsymbol 1  \boldsymbol 1\tr  \theta_{f,J_f0}^{-1})\right \}_{f = 1}^{F},
\end{equation*}
where for each $f = 1, \dots, F$,  $\mathcal{S}_f(\boldsymbol p(\boldsymbol{\theta}_0)) =  \boldsymbol{1} \tr  \mathbf{M} \boldsymbol{p}(\boldsymbol{\theta}_0)$ is the exposure size at block $f$.
\item The asymptotic covariance matrix  $\boldsymbol{\Phi}_{\tilde{\boldsymbol{\theta}}}$ of $\sqrt{N}\hat{\tilde{\boldsymbol{\theta}}}$ is  block-diagonal, with blocks $\boldsymbol{\Phi}_{\tilde{\boldsymbol{\theta}}_{f}}$,  where for each for $f = 1, \dots, F$, 
\begin{equation}\label{AsCovMultiFloret}
\boldsymbol{\Phi}_{\tilde{\boldsymbol{\theta}}_f} =\frac{1}{ {\mathcal{S}_f(p(\boldsymbol{\theta}_0))}}  (diag(\tilde{\boldsymbol \theta}_{f0})  - \boldsymbol 1  \boldsymbol 1\tr  \theta_{f,J_f,0}). 
\end{equation}
\end{enumerate}
\end{theorem}

\begin{proof}

The statement (\textit{i}) is verified in  the Appendix. 

To show (\textit{ii}) recall that the inverse of a non-singular block-diagonal matrix is also block-diagonal,  obtained by inverting the original blocks.  It is straightforward to check that each block of $\mathbf{A}\tr \mathbf{A}$ is non-singular,  with inverse equal to:
$$
(\mathcal{S}_f(\boldsymbol p(\boldsymbol{\theta}_0))( diag(\tilde{\boldsymbol \theta}_{f0}^{-1})  + \boldsymbol 1  \boldsymbol 1\tr  \theta_{f,J_f0}^{-1}))^{-1} = \frac{1}{ {\mathcal{S}(p(\boldsymbol{\theta}_0))}}  (diag(\tilde{\boldsymbol \theta}_{f0})  - \boldsymbol 1  \boldsymbol 1\tr  \theta_{f,J_f,0}),$$
for $f = 1, \dots, F$.
\end{proof}

Similarly to the one-floret case,  each $\boldsymbol{\Phi}_{\tilde{\boldsymbol{\theta}}_f}$ in (\ref{AsCovMultiFloret}) can be written as a product 
$$\boldsymbol{\Phi}_{\tilde{\boldsymbol{\theta}}_f} = \frac{N}{ {\mathcal{S}_f(Np(\boldsymbol{\theta}_0))}}  \boldsymbol{\Sigma}_{\tilde{\boldsymbol{\theta}}_f},$$ 
of the covariance matrix of ${\tilde{\boldsymbol{\theta}}}_f$,  $\boldsymbol{\Sigma}_{\tilde{\boldsymbol{\theta}}_f} = diag(\tilde{\boldsymbol{\theta}}_f) - \boldsymbol 1  \boldsymbol 1\tr \theta_{f,J_f}$,   
and the reciprocal of the floret-specific exposure rate,  $\frac{N}{ {\mathcal{S}_f(Np(\boldsymbol{\theta}_0))}} $.  

\vspace{4mm}

\noindent \textit{Comment}:
The asymptotic covariance matrix $\boldsymbol{\Phi}_{\tilde{\boldsymbol{\theta}}}$ can be derived using the  approach of  \cite{AitchSilvey60} for the maximum likelihood maximization under constraints.  Applying their method, one first arrives to the expression for the asymptotic covariance matrix of $\hat{\boldsymbol \theta} = ((\hat{\tilde{\boldsymbol \theta}}_f, \hat{\theta}_{f, J_f}))_{f \in \mathcal{F}}$, as
\begin{equation}\label{AsCovMulti0}
 \boldsymbol{\Phi}_{\boldsymbol{\theta}}  = \mathbf{B}_{\boldsymbol{\theta}_0}^{-1} - \mathbf{B}_{\boldsymbol{\theta}_0}^{-1}  \boldsymbol 1 (\boldsymbol 1\tr\mathbf{B}_{\boldsymbol{\theta}_0}^{-1} \boldsymbol 1)^{-1} \boldsymbol 1\tr\mathbf{B}_{\boldsymbol{\theta}_0}^{-1},
\end{equation}
where $\mathbf{B}_{\boldsymbol{\theta}_0}$ is a block-diagonal matrix consisting of floret-specific blocks  $ {\mathcal{S}_f(p(\boldsymbol{\theta}_{0}))} diag (\boldsymbol{\theta}_{f0}^{-1})$.
Then, a direct substitution of the expression for $\mathbf{B}_{\boldsymbol{\theta}_0}$  into (\ref{AsCovMulti0}) yields  that $\boldsymbol{\Phi}_{{\boldsymbol{\theta}}}$ is  a block-diagonal matrix, with blocks $\boldsymbol{\Phi}_{{\boldsymbol{\theta}}_{f}}$, for $f = 1, \dots, F$, 
$$\boldsymbol{\Phi}_{{\boldsymbol{\theta}}_f} = \frac{N}{ {\mathcal{S}_f(Np(\boldsymbol{\theta}_0))}}  \boldsymbol{\Sigma}_{{\boldsymbol{\theta}}_f},$$ 
where $\boldsymbol{\Sigma}_{{\boldsymbol{\theta}}_f} = diag({\boldsymbol{\theta}}_f{0}) - \boldsymbol \theta_{f0}  \boldsymbol \theta_{f0}\tr$ is the covariance matrix of $\boldsymbol \theta_{f0}$.
Finally,  using the definition of $\tilde{\boldsymbol{\theta}}_{f0}$ in terms of ${\boldsymbol{\theta}}_{f}$, the delta method yields the expression (\ref{AsCovMultiFloret}). \qed

\vspace{4mm}

Notice that for a multi-floret model,   the true  exposure rates and their expected values are floret-specific,  equal to the ratio between the  exposure  size (number of subjects involved in the experiment) and the number of times the floret treatment is (or expected to be) applied,  as illustrated next.

\begin{example} \label{ExTwoFlorets1}
Consider a (hypothetical) clinical study in which the efficacy of a therapy consisting of treatments $T_1$ and $T_2$ is investigated.  First,  treatment $T_1$ is administered to all patients.  Then, the patients who had a positive response to it, received $T_1$ for the second time.  The patients who either had a negative response to the first round of $T_1$ or  a positive response to the second round of $T_1$ are treated with $T_2$.  The data-generating procedure is sequential and under the assumption that neither the efficacy of $T_1$ nor the efficacy of $T_2$  depend on how many times $T_1$ was  received by an individual, can be described using the tree in Figure \ref{TreeTwoFlorets}. Let $p_{i}$, $i = 1, \dots, 5$, denote the leaf probabilities corresponding to the five resulting therapy paths.  The non-leaf nodes are partitioned in two equivalence classes, the first one parameterized by  
$(\zeta_1,\zeta_2)$ and the second one by $(\eta_1, \eta_2)$, where  $\zeta_1,\zeta_2, \eta_1, \eta_2 \in (0,1)$ and satisfy $\zeta_1 + \zeta_2 = 1$,  $\eta_1 +\eta_2 = 1$.   The corresponding model  $\mathcal{M}$ has two florets and states that:
\begin{equation}\label{modelTwoFlorets1}
p_{1}= \zeta_1^2\eta_1, \,\,p_{2}=\zeta_1^2\eta_2, \,\,p_{3}= \zeta_1\zeta_2, \,\,p_{4}=\zeta_2\eta_1, \,\, p_{5} =\zeta_2\eta_2.
\end{equation} 
 The model matrix consists of two floret-specific blocks:
$$\mathbf{M} = \left(\begin{array}{rrrrrr}   2 & 2 & 1 & 0 & 0\\ 
                                                             0 & 0 & 1 & 1 & 1 \\
                                                             &  &  &  &  \\[-6pt]
                                                             1 & 0 & 0 & 1 & 0 \\
                                                             0 & 1 & 0 & 0 & 1 \\
                                \end{array} \right).
$$
Notice that the row space of the matrix block corresponding to the second floret does not contain the row $(1,1,1,1,1)$, so the floret is incomplete and the exposure ratio in this floret depends on the observed data. 
By Theorem \ref{theoremMatrixProbMultiFloret}, the asymptotic covariance matrix of the non-redundant $\sqrt{N}\hat{\zeta}_1$ and $\sqrt{N}\hat{\eta}_1$  consists of two blocks and is equal to
\begin{equation} \label{MxFinal}
\boldsymbol{\Phi}_{\zeta_1, \eta_1}= 
 \left(\begin{array}{cc} \frac{\zeta_1(1-\zeta_1)}{\zeta_1 +1} & 0\\
                                 0 & \frac{\eta_1(1-\eta_1)}{\zeta_1^2 - \zeta_1 +1}\\                                
        \end{array}\right).\end{equation}
Indeed, in expectation, the $T_1$-exposure size is $N(\zeta_1 +1)$, because $N$ patients receive $T_1$ at the beginning and  another $N\zeta_1$ are expected to receive $T_1$ at the second round. And the expected $T_2$-exposure size is  $N(\zeta_1^2 + (1- \zeta_1))$, where  $N(1-\zeta_1)$ patients are expected to have had a negative response to the first round of $T_1$ and $N\zeta_1^2$ patients are expected to have had two positive responses to $T_1$.    \qed
\end{example}

\begin{figure}
\begin{center}
\includegraphics[scale=0.9]{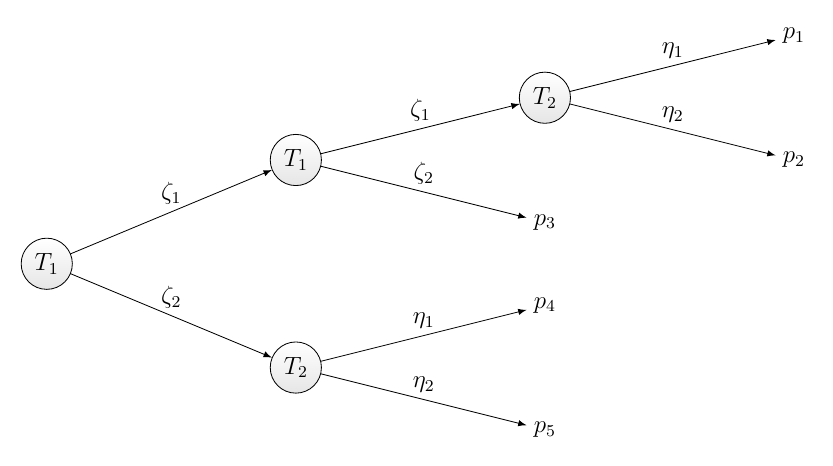}
\end{center}
\caption{Two-treatment regimen conditional on the prior response and assuming independence (no effect) from prior applications in  Example \ref{ExTwoFlorets1}.}
\label{TreeTwoFlorets}
\end{figure}

\section{Examples of experimental design of an incomplete tree structure}\label{SectionExample2}

As the first illustration, consider the example of \cite{JohnsonMay1995}, describing a hypothetical investigation of a two-step treatment regimen for treating patients with a certain disease.  Three disease grades are distinguished, mild, moderate, or severe.  Initially, the treatment is administered  to all patients. After the treatment outcome is evaluated, the patients for whom no improvement was observed,  received the same treatment for the second time. 

 \cite{JohnsonMay1995} assumed the subpopulations based on disease severity were independent and  treated the data as a collection of three stratified triangular $2 \times 2$ contingency tables.  The approach described here  differs from the one of  \cite{JohnsonMay1995}.  The data collection procedure is viewed as sequential, leading to the experimental design shown in Figure \ref{TreeJohnson}.  The corresponding data analysis is presented next.  
 
Let $ \eta_1, \eta_2, \eta_3 \in (0,1) $ denote the probabilities of having, respectively,  mild, moderate, and severe disease; $ \eta_1 + \eta_2 + \eta_3 = 1$.  Further,  let $\zeta_1 \in (0,1) $ be probability of no improvement in a patient after the first treatment round, and $\zeta_2 =  1 - \zeta_1$. The hypothesis to be tested is that the chance of improvement is independent of the initial disease severity and of the treatment step, which can be expressed using a multiplicative model  parameterized as:
\begin{align}\label{design2}
p_{1} &=\eta_1\zeta_1^2, \,\,p_{2}=\eta_1\zeta_1 \zeta_2, \,\,p_{3}= \eta_1\zeta_2,  \nonumber \\
p_{4} &=\eta_2\zeta_1^2, \,\,p_{5}=\eta_2\zeta_1 \zeta_2, \,\,p_{6}= \eta_2\zeta_2, \\
p_{7} &=\eta_3\zeta_1^2, \,\,p_{8}=\eta_3\zeta_1 \zeta_2 \,\,p_{9}= \eta_3\zeta_2. \nonumber
\end{align} 
The model can also be expressed as $\boldsymbol p = \boldsymbol \theta^{\mathbf{M}\tr}$,  with the model matrix equal to
\begin{equation} \label{calves3}
\mathbf{M} = \left(\begin{array}{ccccccccc} 1 & 1 & 1 & 0 & 0 & 0 & 0 & 0 &0  \\ 
                                                                0 & 0 & 0 & 1 & 1 & 1 & 0 & 0 & 0  \\ 
                                                                0 & 0 & 0 & 0 & 0 & 0 & 1 & 1 & 1  \\ 
                                                                &&&&&&&&\\  [-6pt]
                                                                2 & 1 & 0 & 2 & 1 & 0 & 2 & 1 & 0  \\ 
                                                                0 & 1 & 1 & 0 & 1 & 1 & 0 & 1 & 1  \\ 
                                                               \end{array} \right),
\end{equation}
where $\boldsymbol p = (p_1, \dots, p_9)\tr$ and $\boldsymbol \theta= (\eta_1, \eta_2, \eta_3, \zeta_1, \zeta_2)\tr$.

\begin{figure}
\begin{center}
\includegraphics[scale=0.6]{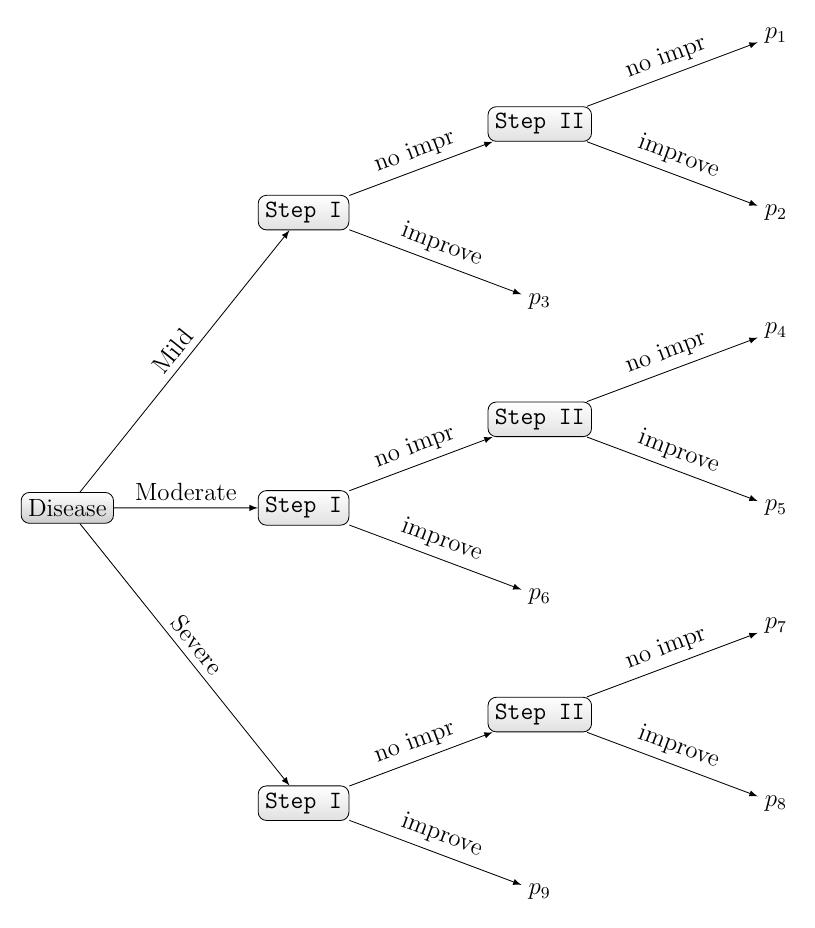}
\end{center}
\caption{A two-step regimen described in \cite{JohnsonMay1995}}
\label{TreeJohnson}
\end{figure}

As implied by the study design and the hypothesis of interest, the model (\ref{design2}) is a two-floret model.  Based on the sampling scheme, one can assume the data $\boldsymbol y$ come from  a multinomial distribution $Mult(N, \boldsymbol p)$.   
By Theorem \ref{theoremMatrixMultiFloret},  the MLEs are
\begin{align}\label{design2thetaMLE}
&\hat{\eta}_1 = \frac{y_1 + y_2 + y_3}{N}, \quad \hat{\eta}_2 = \frac{y_4 + y_ 5 + y_6}{N}, \quad \hat{\eta}_3 = \frac{y_7 + y_ 8 + y_9}{N}, \\ \nonumber
&\hat{\zeta}_1 = \frac{2y_1 + y_2 + 2y_4 + y_ 5 + 2y_7 + y_ 8}{2y_1 + 2y_2 + y_3 + 2y_4 +   2y_ 5 + y_6 + 2y_7 + 2y_ 8 + y_9}, \\  
&\hat{\zeta}_2 = \frac{y_2 + y_3 +  y_ 5 + y_6 +  y_ 8 + y_9}{2y_1 + 2y_2 + y_3 + 2y_4 +   2y_ 5 + y_6 + 2y_7 + 2y_ 8 + y_9}. \nonumber
\end{align}
The asymptotic covariance matrix of the non-redundant $\sqrt{N}(\hat{\eta}_1, \hat{\eta}_2)\tr$ and $\sqrt{N}\hat{\zeta}_1$  consists of two blocks: 
\begin{equation*} 
\boldsymbol{\Phi}_{{\eta}_1,  \eta_2, {\zeta}_1}= 
 \left(\begin{array}{ccc} \frac{\eta_1(1-\eta_1)}{1} & -\eta_1\eta_2 & 0\\
                                        -\eta_2\eta_1 & \frac{\eta_2(1-\eta_2)}{1} & 0 \\
                                 0 & 0 & \frac{\zeta_1(1-\zeta_1)}{\zeta_1 +1}\\                                
        \end{array}\right).\end{equation*}
For the data in Table IV of \cite{JohnsonMay1995}, ordered according to the tree  in Figure \ref{TreeJohnson},  $\boldsymbol y =$ $(46$, $83,$ $176,$ $16,$ $37, $91,$ $6,  $21, $ $43)\tr$,   the MLEs under (\ref{design2}) are $\hat {\eta}_1 = \frac{305}{519}$, $ \hat {\eta}_2 =  \frac{144}{519}$, $ \hat {\zeta}_1 = \frac{277}{728}$,
$\hat{\boldsymbol y} = N \hat{\boldsymbol p}  \approx $ $(44.16, $ $71.89,$ $188.95,$  $20.85,$  $33.94,$  $89.21,$  $10.13,$  $16.50,$  $43.37)\tr$. The goodness-of-fit statistics are  $X^2 = 7.04$, $G^2 = 7.26$ on $df = 9 - (2 + 1) -1 = 5$ degrees of freedom,  indicating a good fit.

As the second illustration,  the data obtained from a clinical study, called DIA-Vacc (NCT04799808), is analysed.  The study  was initiated in January 2021 in several nephrology centres in  Germany \citep{DIAVacc}.  At the beginning, the participants were given two doses either of BNT162b2mRNA or 1273-mRNA vaccine, and those whose vaccine-specific antibodies reached a certain threshold  were said to have a positive immune response to the vaccination.  After six months, the participants without a positive response to the inital vaccination were administered a \textit{booster} of the same vaccine type.  After another several months, the participants  who did not respond to the first booster,  were given another one.  Figure \ref{TreeVaccine}  visualizes the applied vaccination scheme.  One of the research goals in DIA-Vacc was to explore the possibility of existence of a delayed immune response, that is,  whether a  response to the first vaccination is independent from responses to the subsequent revaccinations.  Denote by $\theta_1, \theta_2 \in (0,1)$ the probabilities of failed and successful immune response after the first vaccination application.  Under the hypothesis of independence,   the probabilities of failure and success remain constant during each subsequent revaccination,  so $\theta_1$ and $\theta_2$ are associated to the corresponding tree edges for each node in Figure \ref{TreeFVaccine}.  The leaf probabilities $\boldsymbol p = $ ($p_1$, $p_2$, $p_3$, $p_4)\tr$ are, therefore, equal to
\begin{equation}\label{treeçonstraints}
p_{1} = \theta_1^3, \,\,  p_{2} = \theta_1^2 \theta_2, \,\, p_{3} = \theta_1\theta_2, \,\, p_{4} = \theta_2.
\end{equation}
The constraints (\ref{treeçonstraints}) specify a one-floret tree model $\boldsymbol p =  \boldsymbol \theta^{\mathbf{M}\tr}$ with the design matrix:
\begin{equation*}
\mathbf{M} = \left( 
\begin{array}{cccc}
3&2&1&0\\
0&1&1&1\\
\end{array}
\right).
\end{equation*}
As follows from Theorem \ref{theoremMatrixOneFloret}, the MLEs  of the edge parameters $\theta_1,  \theta_2$ are:
\begin{align}\label{VaccineMLE}
&\hat{\theta}_1 = \frac{3y_1 + 2y_2 + y_3}{3y_1 + 3y_2 + 2y_3 + y_4 },  \quad 
\hat{\theta}_2= \frac{y_2 + y_3 + y_4}{3y_1 + 3y_2 + 2y_3 + y_4 },
\end{align}
and,  by Theorem \ref{theoremMatrixProbOneFloret},  the asymptotic variance of $\sqrt{N}\hat{\theta}_1$ is equal to 
$$\boldsymbol{\Phi}_{{\theta}_1} = {\theta_1(1-\theta_1)}/{(\theta _1^2 + \theta _1+ 1)}.$$

To limit the disclosure, only the data of kidney transplant recipients participating in the study collected during an interim analysis after three vaccination rounds will be considered.   The set frequencies for each respective vaccination path $\boldsymbol y$ = $(80,$ $12,$ $44,$ $64)\tr$ was obtained by sampling from a multinomial distribution based on actually observed interim relative frequencies.  For these data, $\hat{\theta} = 308/428 \approx 0.72$ and 
$\hat{\boldsymbol p}$ $\approx (0.373,$
$0.145,$ $0.202,$ $0.280)\tr$.  The Pearson statistic $X^2 \approx 11.85$  and deviance $G^2 \approx 14.65$, on two degrees of freedom, indicate evidence against the hypothesis of independence.

\section{Discussion and practical implications}\label{sectionDiscussion}

The paper described a framework for modelling the data obtained from sequential experiments with categorical responses and focused primarily on the designs that do not lead to data with a full Cartesian product sample space. Such incomplete tree designs occur naturally in several applications, see, e.g., Section \ref{SectionExample2}. The framework is built upon staged-tree models that are widely discussed in the field of algebraic statistics. 

The contribution of this paper is twofold. Firstly, the paper shows how certain sequential tree-type hypotheses can be expressed and analysed using the tools available for discrete exponential families and general log-linear models and gives a complete characterisation of the MLE under such hypotheses.
The mean-value theorem proved here gives a complete characterisation of the MLE, thus establishing what information contained in the observed data is retained in the MLE. The asymptotic covariance matrix of the MLE is obtained as a closed expression, and thus, in practice, its computation would not involve any numerical algorithms.  
Secondly, the paper introduces the concept of the exposure ratio serving as an adjustment factor connecting the observed sufficient statistics to their MLEs. If the design is incomplete, that is, some of the floret specific design matrices do not have an overall effect, the exposure ratio may depend on the observed data. In this case,   despite a prespecified initial sample size, the number of observations, or interventions, that would in reality be obtained, respectively administered,  cannot be determined in advance. 

The results of the paper have far-reaching consequences for planning such experiments and for the interpretation of their results. Initial sample size calculations aiming at bounded variances or specific power need to use the correct variance formulas. The same applies in the interpretation of the results for confidence interval  calculations. As shown in the paper, the correct variance estimate uses, instead of the observed sample size, the expected sample size.

To illustrate the potential significance of this, some numerical results are given. In this hypothetical example, the response rate for treatment $T_2$ is investigated for such patients, who failed to give a positive response to another treatment, $T_1$, on each of three consecutive applications of it. Assume, to keep the example simple, that the probability of negative response to $T_1$ does not depend on previous responses to $T_1$, and it is, say, $p=0.3$ in each of the first three rounds of the experiment.  The experiment starts with a sample of, say, $n=200$ individuals. Let $a$, $b$, and $c$ be the respective numbers of those who gave negative responses to $T_11$ in the first three rounds. Of course, $A \sim B(200, 0.3)$, $B \sim B(a, 0.3)$, and $C \sim B(b,0.3)$. Then, the observed exposure size for the fourth round of the experiment, that is to the application of $T_2$, is $c$. The expected exposure size is calculated with the MLE of the negative response probability $p$, $\hat{p}=({a+b+c})/({n+a+b})$,
as $n \hat{p}^3.$
The proportional change of the estimated variance of the success rate of $T_2$ due to the adjustment is proportional to the ratio of the expected to the observed exposure sizes, that is, $R= n \hat{p}^3 / c.$

In a small simulation with $10.000$ replications, the fraction of cases when $R$ was outside of  $(0.8,  1.2)$ was approximately  $0.55$, even if one disregarded the $47$ cases when $c=0$ was observed.  In the same simulation, the ratio of cases when $R$ was outside $(0.6, 1.4)$ was about $0.23$. In more than half of the cases, the variance adjustment was above $20\%$ and in nearly one quarter of the cases, the adjustment was above $40\%$.

 Of course, for larger initial sample sizes one can expect the adjustment to become less important. But even with $n=1500$, more than $20\%$ adjustment occurred in about $11\%$ of the simulated data sets.

\begin{figure}
\begin{center}
\includegraphics[scale=0.9]{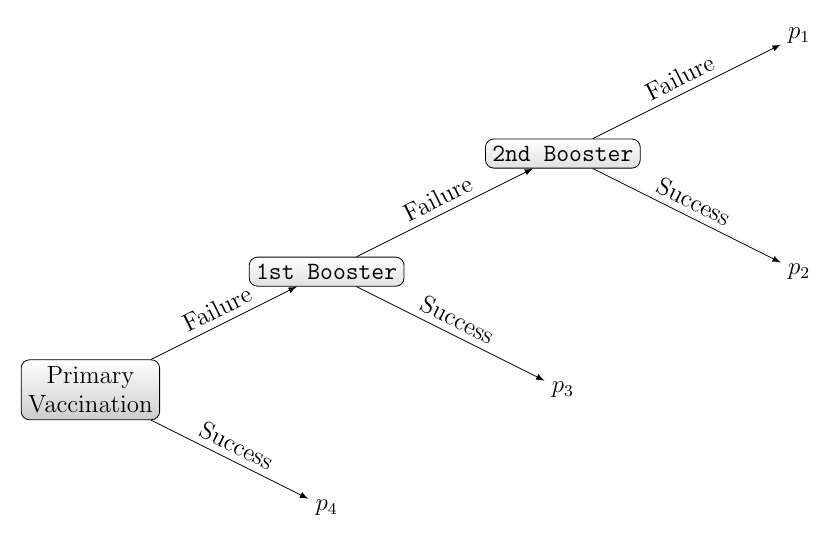}
\end{center}
\caption{Experimental design for a subsequent vaccination procedure.}
\label{TreeVaccine}
\end{figure}

\begin{figure}
\begin{center}
\includegraphics[scale=0.9]{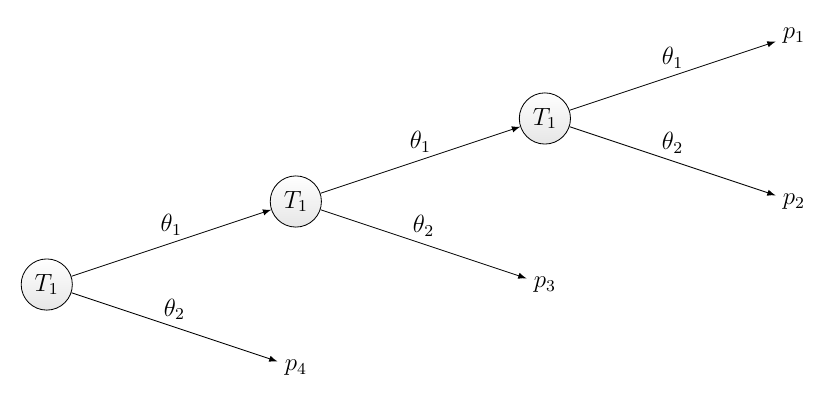}
\end{center}
\caption{Subsequent vaccination conditional on immune response under the assumption of independence.}
\label{TreeFVaccine}
\end{figure}

\section*{Acknowledgement}

The DIA-Vacc data is courtesy of Christian Hugo and Julian Stumpf, from the University Clinic Carl-Gustav Carus, Technical University Dresden. The authors are grateful to Ren\'{e} Mauer for retrieving it. 


\section*{Appendix}

\subsection*{Proof of Theorem \ref{theoremMatrixProbOneFloret}}

\begin{proof}

Part $(i)$: 

For simplicity of presentation, the subscript $0$ in $\boldsymbol p_0$ and $\boldsymbol \theta_0$ is omitted. The derivatives are with respect to $\tilde{\boldsymbol \theta}$, although the notation $\boldsymbol \theta$ is used instead.

Firstly, because $\mathbf{A} = diag(\boldsymbol p)^{-1/2 } \frac{\partial \boldsymbol p}{\partial \boldsymbol{\theta}}$, one has:
$$\mathbf{A}\tr \mathbf{A} = \frac{\partial \boldsymbol p\tr}{\partial \boldsymbol{\theta}} diag(\boldsymbol p^{-1 }) \frac{\partial \boldsymbol p}{\partial \boldsymbol{\theta}}.$$
Thus, to prove (\ref{OneFloretMxA}), one needs to verify that 
\begin{equation}\label{EqProof}
\frac{\partial \boldsymbol p\tr}{\partial \boldsymbol{\theta}} diag(\boldsymbol p^{-1 }) \frac{\partial \boldsymbol p}{\partial \boldsymbol{\theta}}= {\mathcal{S}(\boldsymbol p(\boldsymbol{\theta}))}(diag(\tilde{\boldsymbol\theta}^{-1} ) + \boldsymbol 1  \boldsymbol 1\tr \theta_{J}^{-1}).
\end{equation}

The proof presented next is based on the fact that if $\boldsymbol p \in \mathcal{M}$, the derivative  $\frac{\partial \boldsymbol p}{\partial \boldsymbol{\theta}}$ can be obtained either by explicit differentiation of 
$\boldsymbol p (\boldsymbol \theta) = \boldsymbol \theta^{\mathbf{M}\tr}$, or using implicit differentiation of $\boldsymbol \theta = \frac{1}{\mathcal{S}(\boldsymbol p(\boldsymbol{\theta}))}\mathbf{M}\boldsymbol p(\boldsymbol \theta)$. By bringing the two approaches together, the result in (\ref{EqProof}) and thus in (\ref{OneFloretMxA}) will be achieved.  

Let $\boldsymbol p$ denote the observed probability distribution, and  $\hat{\boldsymbol p} = (\hat{p}_{i})_{i = 1}^I$ the MLE having observed $\boldsymbol p$ under $\mathcal{M}$. Then, because $\mathcal{M}$ is a one-floret tree model,  
\begin{equation}\label{MLEproof}
\hat{p}_i  =  \prod_{j = 1 } ^J \left(\frac{\sum_{i = 1}^I \mu_{ji} p_i}{\sum_{l = 1 }^J \sum_{i = 1}^I \mu_{li} p_i}\right)^{ \mu_{ji}},  \quad i = 1, \dots, I. \tag{*}
\end{equation}
Here $(\mu_{ji})_{j,i}$ are the corresponding entries of the model matrix $\mathbf{M}$. Letting 
$z_j = \sum_{i = 1}^I \mu_{ji} p_i$, for each $j = 1, \dots, J$, $\boldsymbol z = (z_1, \dots, z_J)\tr $, $z_{J + 1} = z_{1} + \dots + z_{J}$, and $\theta_j = z_{j} / z_{J + 1}$, for $j = 1, \dots, J$, one obtains that 
\begin{equation}\label{thetaproof}
\hat{p}_{i} = \prod_{j = 1}^J \left(\frac{z_{j}}{z_{J + 1}}\right)^{\mu_{ji}} = \prod_{j = 1}^J \theta_j^{\mu_{ji}},  \quad i = 1, \dots, I.
\end{equation}

Define
\begin{equation}\label{Qproof1}
Q_p = \sum_{i = 1}^I p_i \log \hat{p}_i. 
\end{equation}
The function $Q_p$ has the following properties. Firstly, by Proposition 2.5 in \cite{Kapranov1991},  for functions of the form (\ref{MLEproof}) one has:
\begin{equation} \nonumber
\log {\hat{p}}_k = \frac{\partial Q_p(\boldsymbol \theta)}{\partial p_k}, \quad k = 1, \dots, I.
\end{equation}
Secondly, after substituting into $Q_p$, the expression for $\hat{p}_i$ in terms of $z_j$'s from (\ref{thetaproof}) one has:
\begin{align*}
Q_p &= \sum_{i = 1}^I p_i \log \prod_{j = 1}^J\left(\frac{z_{j}}{z_{J + 1}}\right)^{\mu_{ji}} 
= \sum_{i = 1}^I p_i \sum_{j = 1}^J {\mu_{ji}}  \log \left(\frac{z_{j}}{z_{J + 1}}\right) \\
&= \sum_{i = 1}^I \sum_{j = 1}^J {\mu_{ji}}  p_i \log {z_{j}} -   \sum_{i = 1}^I  \sum_{j = 1}^J p_i {\mu_{ji}} \log {z_{J + 1}}\\
&=  \sum_{j = 1}^J  \log {z_{j}}\sum_{i = 1}^I{\mu_{ji}}  p_i -  \log {z_{J + 1}} \sum_{j = 1}^J \sum_{i = 1}^I  p_i {\mu_{ji}} \\
&=  \sum_{j = 1}^J z_{j}  \log {z_{j}}-  \log {z_{J + 1}} \sum_{j = 1}^J z_{j}  
=  \sum_{j = 1}^J z_{j}  \log {z_{j}}- {z_{J + 1}} \log {z_{J + 1}}.
\end{align*}
Denote the final expression by 
\begin{equation}\label{Qproof2}
W_z = \sum_{j = 1}^J z_j \log z_j - z_{J+1} \log z_{J + 1}.
\end{equation}

Therefore, for a $\boldsymbol p \in \mathcal{M}$,
\begin{equation}\label{qbothproof}
Q_p(\boldsymbol \theta) = W_z(\boldsymbol \theta).%
\end{equation}
Next, both sides of the equality (\ref{qbothproof}) will be differentiated twice with respect to $\boldsymbol \theta$,  the resulting second derivatives will be set equal to each other, and then the latter will allow to determine the expression for $\frac{\partial \boldsymbol p\tr}{\partial \boldsymbol{\theta}} diag(\boldsymbol p^{-1 }) \frac{\partial \boldsymbol p}{\partial \boldsymbol{\theta}}$.

First, differentiate $Q_p$ with respect to $\boldsymbol \theta$:

\begin{align*}
\frac{\partial Q_p}{\partial \boldsymbol{\theta}} &= \sum_{i = 1}^{I} \frac{\partial p_i}{\partial \boldsymbol{\theta}}  \log \hat{p}_i + \sum_{i = 1}^{I} p_i \frac{\partial p_i/ \partial \boldsymbol{\theta}}{  \hat{p}_i}
= \sum_{i = 1}^{I} \frac{\partial p_i}{\partial \boldsymbol{\theta}}  \log \hat{p}_i,
\end{align*}
because if $\boldsymbol p \in \mathcal{M}$,   $\boldsymbol p = \hat{\boldsymbol{p}}$ and $\sum_{i = 1}^{I}  {\partial p_i/ \partial \boldsymbol{\theta}} = 0$.

The second derivative of $Q_p$ with respect to $\boldsymbol \theta$ is, therefore,
\begin{align*}
\frac{\partial^2 Q_p}{\partial \boldsymbol{\theta}^2} &= \sum_{i = 1}^{I} \frac{\partial^2 p_i}{\partial \boldsymbol{\theta}^2}  \log \hat{p}_i  + \sum_{i = 1}^{I}  \frac{\partial p_i}{\partial \boldsymbol{\theta}} \frac{\partial p_i/ \partial \boldsymbol{\theta}}{  \hat{p}_i},
\end{align*}
or, written in the matrix form,
\begin{align*}
\frac{\partial^2 Q_p}{\partial \boldsymbol{\theta}^2} &=  \frac{\partial^2 \boldsymbol p\tr}{\partial \boldsymbol{\theta}^2}  \log \hat{\boldsymbol p}  + \frac{\partial \boldsymbol p\tr}{\partial \boldsymbol{\theta}} diag(\hat{\boldsymbol p}^{-1 }) \frac{\partial \boldsymbol p}{\partial \boldsymbol{\theta}},
\end{align*}
and, finally, after recalling that $\boldsymbol p = \hat{\boldsymbol p}$:
\begin{align*}
\frac{\partial^2 Q_p}{\partial \boldsymbol{\theta}^2} &= \frac{\partial^2 \boldsymbol p\tr}{\partial \boldsymbol{\theta}^2}  \log \boldsymbol p  + \frac{\partial \boldsymbol p\tr}{\partial \boldsymbol{\theta}} diag(\boldsymbol p^{-1 }) \frac{\partial \boldsymbol p}{\partial \boldsymbol{\theta}}.
\end{align*}

Next, differentiate $W_z$ with respect to $\theta_k$, for $k = 1, \dots, J-1$:
\begin{align*}
\frac{\partial W_z}{\partial {\theta_k}} &= \sum_{j = 1}^{J} \frac{\partial z_j}{\partial {\theta_k}}  \log z_j + \sum_{j = 1}^{J} z_j \frac{\partial z_j/ \partial {\theta_k}}{  z_j} - \frac{\partial z_{J+1}}{\partial {\theta_k}}  \log z_{J+1} - z_{J+1}\frac{\partial z_{J+1}/ \partial {\theta_k}}{  z_{J+1}}
\end{align*}
Because $z_{J + 1} = \sum_{j = 1}^{J} z_{j}$ and thus, $\sum_{j = 1}^{J}  {\partial z_j/ \partial \boldsymbol{\theta}} - \partial z_{J+1}/ \partial \boldsymbol{\theta}= 0$, the expression for the derivative simplifies to:
\begin{align*}
\frac{\partial W_z}{\partial {\theta_k}} &= \sum_{j = 1}^{J} \frac{\partial z_j}{\partial {\theta_k}}  \log z_j  - \frac{\partial z_{J+1}}{\partial {\theta_k}}  \log z_{J+1}.
\end{align*}

Compute the second derivatives of $W_z$ and rearrange the terms:
\begin{align*}
\frac{\partial^2 W_z}{\partial {\theta_k \theta_m}} &= \sum_{j = 1}^{J} \frac{\partial^2 z_j}{\partial {\theta_k \theta_m}}  \log z_j  + \sum_{j = 1}^{J}  \frac{\partial z_j}{\partial{\theta_k}} \frac{\partial z_j/ \partial {\theta_m}}{ z_j} -  \frac{\partial^2 z_{J+1}}{\partial {\theta_k \theta_m}}  \log z_{J +1} - \frac{\partial z_{J +1}}{\partial {\theta_k}} \frac{\partial z_{J +1}/ \partial {\theta_m}}{ z_{J +1}}  \\
&=\sum_{j = 1}^{J} \frac{\partial^2 z_j}{\partial {\theta_k \theta_m}}  \log z_j   -  \frac{\partial^2 z_{J+1}}{\partial {\theta_k \theta_m}}  \log z_{J +1} + \sum_{j = 1}^{J}  \frac{\partial z_j}{\partial {\theta_k}} \frac{\partial z_j/ \partial {\theta_m}}{ z_j} - \frac{\partial z_{J +1}}{\partial {\theta_k}} \frac{\partial z_{J +1}/ \partial {\theta_m}}{ z_{J +1}},  \\
&=\sum_{j = 1}^{J} \frac{\partial^2 z_j}{\partial {\theta_k}^2}  \log (z_j/z_{J +1})  + \sum_{j = 1}^{J}  \frac{\partial z_j}{\partial {\theta_k}} \frac{\partial z_j/ \partial {\theta_k}}{ z_j} - \frac{\partial z_{J +1}}{\partial {\theta_k}} \frac{\partial z_{J +1}/ \partial {\theta_k}}{ z_{J +1}}, \\
\end{align*}
for $ m = 1,\dots, J-1$. 
To simplify further, recall that $z_j = \theta_j z_{J+1}$, for $j = 1, \dots, J$ and $\theta_J = 1 - \sum_{j = 1}^{J-1} \theta_j$, and take into account that for $k = 1, \dots, J -1$, $$\partial z_{J +1}/ \partial {\theta_k} = \sum_{j = 1}^J \partial z_{j}/ \partial {\theta_k} = \partial z_{k}/ \partial {\theta_k} + \partial z_{J}/ \partial {\theta_k} = z_{J+1} - z_{J+1} = 0.$$
 
\noindent For $k = m$, one has:
\begin{align*}
\frac{\partial^2 W_z}{\partial {\theta_k}^2} &=\sum_{j = 1}^{J} \frac{\partial^2 z_j}{\partial {\theta_k}^2}  \log (z_j/z_{J +1})  + \sum_{j = 1}^{J}  \frac{\partial z_j}{\partial {\theta_k}} \frac{\partial z_j/ \partial {\theta_k}}{ z_j} - \frac{\partial z_{J +1}}{\partial {\theta_k}} \frac{\partial z_{J +1}/ \partial {\theta_k}}{ z_{J +1}}  \\
&=\sum_{j = 1}^{J} \frac{\partial^2 z_j}{\partial {\theta_k}^2}  \log (z_j/z_{J +1}) + \frac{\partial z_k}{\partial {\theta_k}} \frac{\partial z_k/ \partial {\theta_k}}{ z_k}  + \frac{\partial z_{J}}{\partial {\theta_k}} \frac{\partial z_J/ \partial {\theta_k}}{ z_J} - \frac{\partial z_{J +1}}{\partial {\theta_k}} \frac{\partial z_{J +1}/ \partial {\theta_k}}{ z_{J +1}}  \\
&=\sum_{j = 1}^{J} \frac{\partial^2 z_j}{\partial {\theta_k}^2} \log (z_j/z_{J +1})  + \frac{ z_{J +1}^2}{ z_{J+1}\theta_k}  + (-z_{J +1}) \frac{-z_{J +1}}{ z_{J +1}\theta_J} - \frac{\partial z_{J +1}}{\partial {\theta_k}} \frac{\partial z_{J +1}/ \partial {\theta_k}}{ z_{J +1}}  \\
&=\sum_{j = 1}^{J} \frac{\partial^2 z_j}{\partial {\theta_k}^2}  \log (z_j/z_{J +1})  + z_{J +1}(\frac{1}{\theta_k}  + \frac{1}{\theta_J}).
\end{align*}
\noindent For $k \neq m$:
\begin{align*}
\frac{\partial^2 W_z}{\partial {\theta_k\theta_m}} &=\sum_{j = 1}^{J} \frac{\partial^2 z_j}{\partial {\theta_k \theta_m}} \log (z_j/z_{J +1}) + \sum_{j = 1}^{J}  \frac{\partial z_j}{\partial {\theta_k}} \frac{\partial z_j/ \partial {\theta_m}}{ z_j} - \frac{\partial z_{J +1}}{\partial {\theta_k}} \frac{\partial z_{J +1}/ \partial {\theta_m}}{ z_{J +1}}  \\
&=\sum_{j = 1}^{J} \frac{\partial^2 z_j}{\partial {\theta_k \theta_m}}  \log (z_j/z_{J +1})  + \frac{\partial z_J}{\partial {\theta_k}} \frac{\partial z_J/ \partial {\theta_m}}{ z_j}  - \frac{\partial z_{J +1}}{\partial {\theta_k}} \frac{\partial z_{J +1}/ \partial {\theta_m}}{ z_{J +1}}  \\
&=\sum_{j = 1}^{J} \frac{\partial^2 z_j}{\partial {\theta_k \theta_m}}  \log (z_j/z_{J +1})  + (-z_{J +1}) \frac{-z_{J +1}}{ z_{J +1}\theta_J}  - \frac{\partial z_{J +1}}{\partial {\theta_k}} \frac{\partial z_{J +1}/ \partial {\theta_m}}{ z_{J +1}}  \\
&=\sum_{j = 1}^{J} \frac{\partial^2 z_j}{\partial {\theta_k\theta_m}}  \log (z_j/z_{J +1})  + z_{J +1}(\frac{1}{\theta_J}).
\end{align*}
Finally,
\begin{equation}\label{secondDerivativesOneFloret}
\frac{\partial^2 W_z}{\partial {\theta_k\theta_m}} = \left\{ \begin{array}{l} 
\sum_{j = 1}^{J} \frac{\partial^2 z_j}{\partial {\theta_k}^2}  \log (z_j/z_{J +1})   + z_{J +1}(\frac{1}{\theta_k}  + \frac{1}{\theta_J}), \,  k = m,\\
{}\\
\sum_{j = 1}^{J} \frac{\partial^2 z_j}{\partial {\theta_k\theta_m}}  \log (z_j/z_{J +1}) + z_{J +1}(\frac{1}{\theta_J}), \, k    \neq m.  \\
                                                                                      \end{array}  \right.
\end{equation}

Equivalently, in the matrix form,
\begin{equation*}
\frac{\partial^2 W_z}{\partial \tilde{\boldsymbol\theta}^2} =  \frac{\partial^2 {\boldsymbol z}}{\partial {\tilde{\boldsymbol \theta}}^2} \log ({\boldsymbol z}/z_{J+1}) + z_{J +1}(diag(\tilde{\boldsymbol\theta}^{-1} ) + \boldsymbol 1  \boldsymbol 1\tr \theta_{J}^{-1}).
\end{equation*}

To complete the proof, we need to show that 
$$ \frac{\partial^2 {\boldsymbol z}}{\partial {\tilde{\boldsymbol \theta}^2}} \log ({\boldsymbol z}/z_{J+1}) = 
\frac{\partial^2 \boldsymbol p\tr}{\partial \boldsymbol{\theta}^2}  \log \hat{\boldsymbol p}.$$

Because $z_j = \sum_{i = 1}^I \mu_{ij} p_i$,  $$ \frac{\partial^2 z_j }{\partial \theta_k \theta_m }  = 
\sum_{i = 1}^I \mu_{ij} \frac{\partial^2 p_i }{\partial \theta_k \theta_m } $$

Recall that $\hat{p}_{i} = \prod_{j = 1}^J \theta_j^{\mu_{ji}} =  \prod_{j = 1}^J (z_j/z_{J+1})^{\mu_{ji}} $. Hence, 
\begin{align*}
\sum_{i = 1}^I \frac{\partial^2 p_i}{\partial \theta_k \theta_m}\log \hat{p}_{i} &= \sum_{i = 1}^I \frac{\partial^2 p_i}{\partial \theta_k \theta_m }\sum_{j = 1}^J{\mu_{ji}} \log (z_j/z_{J+1}) = 
\sum_{j = 1}^J  \log (z_j/z_{J+1}) \sum_{i = 1}^I{\mu_{ji}} \frac{\partial^2 p_i}{\partial \theta_k \theta_m}\\
& = \sum_{j = 1}^{J} \log (z_j/z_{J+1}) \frac{\partial^2 z_j}{\partial \theta_k \theta_m}.
\end{align*}

Therefore,
$$ \frac{\partial \boldsymbol p\tr}{\partial \boldsymbol{\theta}} diag(\boldsymbol p^{-1 }) \frac{\partial \boldsymbol p}{\partial \boldsymbol{\theta}} = z_{J + 1}(diag(\tilde{\boldsymbol\theta}^{-1} ) + \boldsymbol 1  \boldsymbol 1\tr \theta_{J}^{-1}) = {\mathcal{S}(p(\boldsymbol{\theta}))}(diag(\tilde{\boldsymbol\theta}^{-1} ) + \boldsymbol 1  \boldsymbol 1\tr \theta_{J}^{-1}),$$
which completes the proof.
\end{proof}

\subsection*{Proof of Theorem \ref{theoremMatrixProbMultiFloret}} 

\begin{proof}

The procedure generalizes  the proof of Theorem \ref{theoremMatrixProbOneFloret}.
The subscript $0$ in $\boldsymbol p_0$ and $\boldsymbol \theta_0$ is also omitted. The derivatives are with respect to $\tilde{\boldsymbol \theta}$, although the notation $\boldsymbol \theta$ is used instead.

Let $\mathbf{A} = diag(\boldsymbol p)^{-1/2 } \frac{\partial \boldsymbol p}{\partial \boldsymbol{\theta}}$. It will be shown first that the matrix 
$$\mathbf{A}\tr \mathbf{A} = \frac{\partial \boldsymbol p\tr}{\partial \boldsymbol{\theta}} diag(\boldsymbol p^{-1 }) \frac{\partial \boldsymbol p}{\partial \boldsymbol{\theta}}$$
is block-diagonal with floret-specific blocks equal to:
\begin{equation*}
\mathbf{B}_f = \mathcal{S}_f(\boldsymbol p(\boldsymbol{\theta}))( diag(\tilde{\boldsymbol \theta_f}^{-1})  + \boldsymbol 1  \boldsymbol 1\tr  \theta_{J}^{-1}), f \in \mathcal{F}.
\end{equation*}

Let $\boldsymbol p$ denote the observed probability distribution, and  $\hat{\boldsymbol p} = (\hat{p}_{i})_{i = 1}^I$ the MLE having observed $\boldsymbol p$ under $\mathcal{M}$. 
Because $\mathcal{M}$ is a multi-floret tree model,  for each $i = 1, \dots, I$, 
\begin{equation}\label{MLEproofMulti}
\hat{p}_i  =   \prod_{f = 1 } ^F\prod_{j = 1 } ^{J_f} \left(\frac{\sum_{i = 1}^I \mu_{fji} p_i}{\sum_{l = 1 }^{J_f} \sum_{i = 1}^I \mu_{fli} p_i}\right)^{ \mu_{fji}}. \tag{**}
\end{equation}
Here $(\mu_{fji})_{f,j,i}$ are the corresponding entries of the model matrix $\mathbf{M}$. For each $f \in \mathcal{F}$, letting 
$z_{fj} = \sum_{i = 1}^I \mu_{fji} p_i$, for each $j = 1, \dots, J_f$, $\boldsymbol z_f = (z_{f1}, \dots, z_{f,J_f})\tr $, $z_{f,J_f + 1} = z_{f1} + \dots + z_{f,J_f}$, and $\theta_{fj} = z_{fj} / z_{f,J_f + 1}$, for $j = 1, \dots, J_f$,  one obtains that 
\begin{equation}\label{thetaproofMulti}
\hat{p}_{i} =\prod_{f = 1 } ^F \prod_{j = 1}^{J_f} \left(\frac{z_{fj}}{z_{f,J_f + 1}}\right)^{\mu_{fji}} = \prod_{f = 1 } ^F\prod_{j = 1}^{J_f} \theta_{fj}^{\mu_{fji}},  \quad i = 1, \dots, I.
\end{equation}

As in the proof of Theorem \ref{theoremMatrixProbOneFloret}, define
\begin{equation}\label{QproofF}
Q_p = \sum_{i = 1}^I p_i \log \hat{p}_i. 
\end{equation}
Generalizing Proposition 2.5 in \cite{Kapranov1991}  for functions of the form (\ref{MLEproofMulti}), one has:
\begin{equation} \nonumber
\log {\hat{p}}_k = \frac{\partial Q_p(\boldsymbol \theta)}{\partial p_k}, \quad k = 1, \dots, I.
\end{equation}
Next, let for $f \in \mathcal{F}$, 
\begin{equation}\nonumber
W_{fz} = \sum_{j = 1}^{J_f} z_{fj} \log z_{fj} - z_{f,J_f+1} \log z_{f,J_f + 1},  \quad W_z = \sum_{f = 1}^F W_z.
\end{equation}
After substituting into $Q_p$ the expression for $\hat{p}_i$ from (\ref{thetaproofMulti}) and rearranging terms one obtains that for a $\boldsymbol p \in \mathcal{M}$,
\begin{equation}\label{qbothproofMulti}
Q_p(\boldsymbol \theta) = W_z(\boldsymbol \theta) = \sum_{f = 1}^F W_{fz}.
\end{equation}

Following the idea that was used in the proof of Theorem  \ref{theoremMatrixProbOneFloret}, the second derivatives with respect to $\boldsymbol \theta$ of the right hand and left hand sides of the equality (\ref{qbothproofMulti}) will be set equal to each other, which would allow to determine the expression for $\frac{\partial \boldsymbol p\tr}{\partial \boldsymbol{\theta}} diag(\boldsymbol p^{-1 }) \frac{\partial \boldsymbol p}{\partial \boldsymbol{\theta}}$.

One can verify that the first and second derivatives of $Q_p$ with respect to $\boldsymbol \theta$ have the same form as in the one-floret case, namely,
\begin{align*}
\frac{\partial Q_p}{\partial \boldsymbol{\theta}} &= \sum_{i = 1}^{I} \frac{\partial p_i}{\partial \boldsymbol{\theta}}  \log \hat{p}_i,\\
\frac{\partial^2 Q_p}{\partial \boldsymbol{\theta}^2} &= \sum_{i = 1}^{I} \frac{\partial^2 p_i}{\partial \boldsymbol{\theta}^2}  \log \hat{p}_i  + \sum_{i = 1}^{I}  \frac{\partial p_i}{\partial \boldsymbol{\theta}} \frac{\partial p_i/ \partial \boldsymbol{\theta}}{  \hat{p}_i}.
\end{align*}
After recalling that if $\boldsymbol p \in \mathcal{M}$ then $\boldsymbol p = \hat{\boldsymbol p}$, one has in the matrix form:
\begin{align*}
\frac{\partial^2 Q_p}{\partial \boldsymbol{\theta}^2} &= \frac{\partial^2 \boldsymbol p\tr}{\partial \boldsymbol{\theta}^2}  \log \boldsymbol p  + \frac{\partial \boldsymbol p\tr}{\partial \boldsymbol{\theta}} diag(\boldsymbol p^{-1 }) \frac{\partial \boldsymbol p}{\partial \boldsymbol{\theta}}.
\end{align*}

Next, notice that for $f, g \in \mathcal{F}$,  for $k = 1, \dots, J_g-1$, one has:
\begin{align*}
\frac{\partial W_{fz}}{\partial {\theta_{gk}}} &= \left \{\begin{array}{ll}
 \sum_{j = 1}^{J_f} \frac{\partial z_{fj}}{\partial {\theta_{fk}}}  \log z_{fj}  - \frac{\partial z_{f,J_f+1}}{\partial {\theta_{fk}}}  \log z_{J_f+1}, &  \,\, f = g, \\
 0 & \,\, f \neq g.\\
 \end{array} \right.
\end{align*}
Therefore,
\begin{align*}
\frac{\partial W_{z}}{\partial {\theta_{fk}}} &= \sum_{g = 1}^F \frac{\partial W_{gz}}{\partial {\theta_{fk}}}  =   \frac{\partial W_{fz}}{\partial {\theta_{fk}}}  = \sum_{j = 1}^{J_f} \frac{\partial z_{fj}}{\partial {\theta_{fk}}}  \log z_{fj}  - \frac{\partial z_{f,J_f+1}}{\partial {\theta_{fk}}}  \log z_{J_f+1}.
\end{align*}

Further, because for the second derivatives,
\begin{align*}
\frac{\partial^2 W_{z}}{\partial {\theta_{fk}}\partial {\theta_{gm}}} &= \frac{\partial^2 W_{fz}}{\partial {\theta_{fk}}\partial {\theta_{fm}}},
\end{align*}
where $k, m = 1, \dots, J_f - 1$, one has:
\begin{equation*}
\frac{\partial^2 W_{z}}{\partial \tilde{\boldsymbol\theta}_f \partial \tilde{\boldsymbol\theta}_g} =
 \left \{\begin{array}{ll}
 \frac{\partial^2 W_{fz}}{\partial \tilde{\boldsymbol\theta}_f^2}&  \,\, f = g, \\
 \mathbf{0} &  \,\, f \neq g.\\
 \end{array} \right.
\end{equation*}
As can be verified using the derivation used in the proof of Theorem  \ref{theoremMatrixProbOneFloret},
\begin{align*}
\frac{\partial^2 W_{fz}}{\partial \tilde{\boldsymbol\theta}_f^2} &=  \frac{\partial^2 {\boldsymbol z}_f}{\partial {\tilde{\boldsymbol \theta}}_f^2} \log ({\boldsymbol z_f}/z_{f,J_f+1}) + z_{f,J_f +1}(diag(\tilde{\boldsymbol\theta_f}^{-1} ) + \boldsymbol 1  \boldsymbol 1\tr \theta_{f,J_f}^{-1}) \\ 
&= \frac{\partial^2 \boldsymbol p\tr}{\partial \boldsymbol{\theta}^2}  \log \hat{\boldsymbol p} + z_{f,J_f +1}(diag(\tilde{\boldsymbol\theta_f}^{-1} ) + \boldsymbol 1  \boldsymbol 1\tr \theta_{f,J_f}^{-1}).
\end{align*}

Therefore,
\begin{align*}
\frac{\partial \boldsymbol p\tr}{\partial \boldsymbol{\theta}} diag(\boldsymbol p^{-1 }) \frac{\partial \boldsymbol p}{\partial \boldsymbol{\theta}} &= diag\left\{z_{f,J_f + 1}(diag(\tilde{\boldsymbol\theta}_f^{-1} ) + \boldsymbol 1  \boldsymbol 1\tr \theta_{f,J_f}^{-1})\right \}_{f = 1}^F \\
&= diag\left\{ {\mathcal{S}_f(p(\boldsymbol{\theta}))}(diag(\tilde{\boldsymbol\theta}_f^{-1} ) + \boldsymbol 1  \boldsymbol 1\tr \theta_{f,J_f}^{-1})\right \}_{f = 1}^F,
\end{align*}

which completes the proof.

\end{proof}

\bibliography{Trees20261018.bib}
\bibliographystyle{apacite}

\end{document}